\documentclass[aps,prxquantum,reprint,floatfix,superscriptaddress,nofootinbib]{revtex4-2} 

\usepackage[utf8]{inputenc}
\usepackage[T1]{fontenc}
\usepackage[english]{babel}
\usepackage{amsmath}
\usepackage{amssymb}
\usepackage{amsthm}
\usepackage{bm}
\usepackage{times}
\usepackage{braket}
\usepackage{graphicx}
\usepackage{mathtools}
\usepackage{tikz}
\usetikzlibrary{arrows.meta, positioning, calc, fit, shapes.geometric}
\usepackage[colorlinks=true,urlcolor=blue,citecolor=blue,linkcolor=blue]{hyperref}

\newtheorem{definition}{Definition}
\newtheorem{proposition}{Proposition}

\newcommand{\PL}{P_{\mathrm{L}}}
\newcommand{\RL}{R}
\newcommand{\SL}{S}
\newcommand{\E}{\mathbb{E}}

\newcommand{\Prob}{\mathbb{P}}
\newcommand{\decoder}{\mathcal{D}}
\newcommand{\circuit}{\mathcal{C}}

\begin{document}

\title{Optimizing quantum error correction through error attribution}

\author{Yuguo Shao}
\affiliation{Yau Mathematical Sciences Center, Tsinghua University, Beijing 100084, China}
\affiliation{Yanqi Lake Beijing Institute of Mathematical Sciences and Applications, Beijing 100407, China}

\author{Zi-Wen Liu}
\affiliation{Yau Mathematical Sciences Center, Tsinghua University, Beijing 100084, China}

\begin{abstract}
Logical error rate is a standard benchmark for quantum error correction, providing an aggregate measure of the overall effect of physical noise. In actual QEC architectures, errors at different circuit locations can have markedly different effects on logical failure even when their probabilities are equal. Here we introduce an {error attribution scheme} to quantitatively characterize the influence of individual circuit components on the logical error rate and use this refined information to guide targeted intervention and design for improved QEC performance. Our attribution estimator computes all component sensitivities jointly from the same Monte Carlo samples used to estimate the logical error rate. The resulting sensitivity maps identify hotspots where targeted noise reduction can most effectively improve logical performance. We test the scheme in circuit simulations of rotated surface code memories and high-rate lifted product codes: halving the noise strength on about 5\%–7\% of physical components selected by sensitivity reduces logical error rates by approximately 15\%–25\%, roughly twice the mean reduction from the same intervention on an equal number of randomly selected components. By connecting performance benchmarking to targeted optimization, our scheme provides a practical strategy for improving quantum error correction with limited resources.
\end{abstract}

\maketitle

\section{Introduction}

Fault-tolerant quantum computation relies on quantum error correction (QEC) to protect quantum information against physical noise and support reliable computation~\cite{kitaev1997quantum,gottesman1998heisenberg,fowler2012surface,litinski2019game}.
Classical simulation provides a widely used tool to assess the feasibility and performance of candidate QEC schemes by estimating logical error rates under specified noise models~\cite{aaronson2004improved,bravyi2016improved,gidney2021stim,higgott2025sparse}.
These estimates help determine error thresholds, compare QEC architectures, and quantify resource requirements, thereby guiding the design of codes, circuits, and decoders for fault-tolerant computation.
However, the logical error rate is only an aggregate measure that reveals little about the mechanisms underlying logical failure and thus offers limited guidance for improving QEC efficiency.

Physical errors associated with different circuit components can have different effects on the logical error rate, even when the component error probabilities are equal.
Previous studies have examined QEC performance under realistic noise~\cite{tomita2014low,gutierrez2016errors,darmawan2017tensor,demarti2022performance} and shown that tailoring codes and decoders to noise characteristics can improve performance~\cite{tiurev2023correcting,higgott2023improved,tuckett2019tailoring,bonilla2021xzzx}.
A complementary approach is to improve a given QEC scheme through targeted reductions in physical noise.
Experimental studies have used noise models informed by device measurements to estimate the contributions of different error sources (such as two-qubit gate errors, idle decoherence, readout errors, leakage, and crosstalk) to logical failure~\cite{acharya2023suppressing,acharya2025belowthreshold}.
To assess interventions on individual circuit components, we need to quantify how reducing each component's noise would affect the overall logical error rate.
Efficiently estimating these effects would help prioritize components for improvement especially when resources for noise reduction are limited.

\begin{figure}[!t]
    \centering
    \includegraphics[width=\columnwidth]{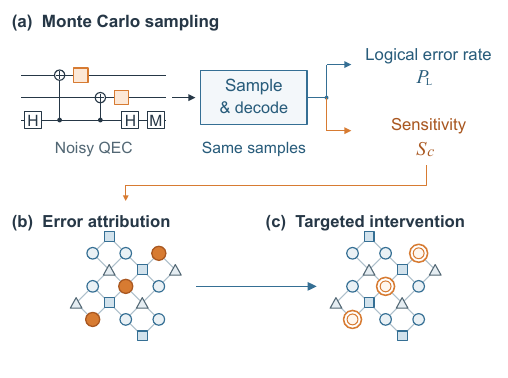}
    \caption{Error attribution framework for optimizing quantum error correction.
    (a) The logical error rate $P_{\mathrm{L}}$ and its sensitivities $S_c$ to component noise strengths are estimated jointly from the same Monte Carlo samples.
    (b) The resulting sensitivity map identifies hotspots (orange), highlighting physical components with large sensitivity magnitudes and thus a strong influence on the logical error rate.
    (c) These hotspots guide targeted noise reduction to improve logical performance under a limited intervention budget.}
    \label{fig:error_attribution_overview}
\end{figure}

In this work, we develop a systematic \emph{error attribution} framework for improving QEC performance by quantitatively mapping the influence of individual circuit components on logical failure, as illustrated in Fig.~\ref{fig:error_attribution_overview}.
We characterize this influence through the sensitivity of the logical error rate to each component's noise strength for fixed QEC circuits and decoders.
To estimate these sensitivities efficiently, we develop an attribution estimator that computes all component sensitivities jointly from the same Monte Carlo samples used to estimate the logical error rate, without separate simulations for individual components or differentiation through the decoder.
The resulting sensitivity maps identify hotspots where changes in component noise have a particularly strong impact on the logical error rate, thereby guiding efficient targeted interventions.

We demonstrate this approach in circuit simulations of rotated surface code memories~\cite{fowler2012surface,gidney2021stim} and high-rate lifted product codes~\cite{mitten2026processors}.
The sensitivity maps reveal substantial variation across physical components.
In the circuits studied here, the selected hotspots predominantly involve data qubit components in the surface code memories and CNOT couplings in the lifted product codes.
By halving the noise strengths of only about $5\%$--$7\%$ of all physical components, these targeted interventions substantially reduce logical error rates by approximately $14\%$--$18\%$ in the surface code memories and $23\%$--$25\%$ in the lifted product codes. These reductions are roughly twice the average reductions obtained with the same intervention on an equal number of randomly selected components, highlighting the substantial benefits of sensitivity-guided interventions for improving QEC performance.

\section{Error attribution map framework}

Attributing the logical error rate to individual noise components requires a precise description of the underlying circuit-level QEC experiment.
Let $\circuit$ denote the circuit implementing a specified memory or logical operation experiment for a chosen code~\cite{gottesman1998heisenberg,fowler2012surface,litinski2019game}.
The circuit specifies the data qubits that encode the logical state, the ancilla qubits used to measure checks, and the scheduled sequence of state preparations or resets, one- and two-qubit gates, idle intervals, and measurements over one or more rounds.

A noisy QEC experiment is specified by the circuit $\circuit$, a stochastic error model, and a decoder $D$.
The error model assigns parameterized error sources to specific circuit elements and spacetime locations, each of which we refer to as a \emph{noise component}.
The measurement record from each run is processed to obtain syndrome information and the readout of a chosen logical observable.
The decoder $D$ uses the syndrome information to infer a correction to this readout.
For an observable with a known ideal outcome, a mismatch between the corrected readout and the ideal outcome constitutes a logical failure.
The overall logical failure rate provides an aggregate measure of the overall effect of physical noise.  Our error attribution framework  is designed to quantify and compare the sensitivity of the overall logical error to each component's noise strength microscopically, thus providing a basis for optimizing QEC performance through informed prioritization, which is especially important with limited calibration or design resources.

We now formalize this error attribution method for a given QEC experiment.
Let $\mathcal{N}=\{1,2,\ldots,e\}$ denote the set of $e$ noise components associated with $\circuit$.
We parameterize the noise model by
\begin{equation}
    \bm{\lambda}
    =
    (\lambda_1,\lambda_2,\ldots,\lambda_e),
\end{equation}
where $\lambda_\alpha$ denotes the noise strength of component $\alpha\in\mathcal{N}$, expressed as an error probability or rate.
For a given $\bm{\lambda}$, each run of the noisy circuit yields two binary records: a syndrome record $\mathbf{m}\in\mathbb{F}_2^d$, supplied to the decoder, and a final readout record $\mathbf{y}\in\mathbb{F}_2^t$, from which the logical readout is obtained.
We denote their joint distribution by $p_{\bm{\lambda}}(\mathbf{m},\mathbf{y})$.

For the chosen logical observable, let $\delta_{\decoder}(\mathbf{m})\in\mathbb{F}_2$ be the correction bit inferred by $\decoder$ from the syndrome $\mathbf{m}$, and let $g(\mathbf{y})\in\mathbb{F}_2$ be the uncorrected readout bit.
Writing $b_{\mathrm{ideal}}\in\mathbb{F}_2$ for the known noiseless readout, the logical failure indicator is
\begin{equation}
    L(\mathbf{m},\mathbf{y},\decoder)
    =
    \mathbf{1}\!\left[
    g(\mathbf{y})
    \oplus
    \delta_{\decoder}(\mathbf{m})
    \neq b_{\mathrm{ideal}}
    \right],
\end{equation}
where $\oplus$ denotes addition in $\mathbb{F}_2$.
More generally, the failure indicator can include an incorrect result for any of several logical observables or a failure of the decoder to return an accepted correction.
The post-decoding logical error rate is the expectation of this indicator over $p_{\bm{\lambda}}(\mathbf{m},\mathbf{y})$:
\begin{equation}
    \PL(\bm{\lambda};\decoder)
    =
    \E_{\bm{\lambda}}\!\left[
    L(\mathbf{m},\mathbf{y},\decoder)
    \right].
    \label{eq:logical_failure_probability}
\end{equation}

Inspired by gradient-based noise-bottleneck diagnostics and experimentally informed error budgets~\cite{shao2025diagnosing,acharya2023suppressing,acharya2025belowthreshold}, we quantify the influence of each noise component by the partial derivative of $\PL$ with respect to its noise strength.
\begin{definition}[Noise sensitivity]
The sensitivity of noise component $\alpha$ is
\begin{equation}
    \SL_{\alpha}
    =
    \frac{\partial \PL(\bm{\lambda};\decoder)}
         {\partial \lambda_{\alpha}}.
    \label{eq:logical_sensitivity}
\end{equation}
\end{definition}
For an independent Bernoulli noise component, this derivative has a direct interpretation in terms of conditional logical failure probabilities.
Let $E_{\alpha}\in\{0,1\}$ be the corresponding error indicator, with $\Prob(E_{\alpha}=1)=\lambda_{\alpha}$.
Conditioning on $E_{\alpha}$ gives
\begin{equation}
    \SL_{\alpha}
    =
    \Prob(L=1\mid E_{\alpha}=1)
    -
    \Prob(L=1\mid E_{\alpha}=0).
    \label{eq:conditional_failure_sensitivity}
\end{equation}

The sensitivity $\SL_{\alpha}$ describes the local response of the logical error rate to a small change in $\lambda_{\alpha}$.
Multiplying this sensitivity by the component's current noise strength gives a linearized measure of its marginal contribution to the logical error rate.
\begin{definition}[Linearized marginal risk]
The linearized marginal risk of noise component $\alpha$ is
\begin{equation}
    \RL_{\alpha}
    =
    \lambda_{\alpha}\SL_{\alpha}
    =
    \lambda_{\alpha}
    \frac{\partial \PL(\bm{\lambda};\decoder)}
         {\partial \lambda_{\alpha}}.
    \label{eq:linearized_marginal_logical_risk}
\end{equation}
\end{definition}
This quantity can be interpreted by comparing the current noise model with one in which component $\alpha$ is removed.
Let $\bm{\lambda}^{(-\alpha)}=(\lambda_1,\ldots,\lambda_{\alpha-1},0,\lambda_{\alpha+1},\ldots,\lambda_e)$ denote the same noise model with component $\alpha$ removed.
A first-order expansion gives
\begin{equation}
    \PL(\bm{\lambda};\decoder)
    -
    \PL(\bm{\lambda}^{(-\alpha)};\decoder)
    \approx
    \RL_{\alpha}.
    \label{eq:linearized_marginal_removal_approximation}
\end{equation}

A physical component, such as a qubit or CNOT coupling, can be used repeatedly in the circuit.
To assess improvements to this component $c$, we account for all associated noise components, indexed by $\mathcal{N}_c\subseteq\mathcal{N}$.
If these noise components share a common error probability $p_c$, so that $\lambda_\alpha=p_c$ for all $\alpha\in\mathcal{N}_c$, the chain rule gives
\begin{equation}
    \SL_c
    =
    \frac{\partial\PL}{\partial p_c}
    =
    \sum_{\alpha\in\mathcal{N}_c}\SL_\alpha.
    \label{eq:physical_component_sensitivity}
\end{equation}
The corresponding linearized marginal risk is
\begin{equation}
    \RL_c
    =
    p_c\SL_c
    =
    \sum_{\alpha\in\mathcal{N}_c}\RL_\alpha.
    \label{eq:physical_component_marginal_risk}
\end{equation}
More generally, noise components can be grouped by gate type, syndrome round, or circuit region to assess interventions affecting multiple components.

\section{Attribution Estimator}

The central computational task is to estimate the
component sensitivities $\SL_\alpha$, from which the  linearized
marginal risks $\RL_\alpha=\lambda_\alpha\SL_\alpha$ follow directly.
We construct an estimator that computes all component sensitivities jointly from the same circuit-level Monte Carlo samples used to estimate $\PL$, without separate perturbed simulations or differentiation through the decoder.
The dominant computational complexity is therefore the same as that of estimating $\PL$ alone.
Throughout this section, we assume that the outcomes of the elementary noise components are mutually independent and that the distribution of component $\alpha$ depends only on its own noise strength $\lambda_\alpha$.

Consider $N$ independent shots.
Shot $i$ produces records $(\mathbf{m}^{(i)},\mathbf{y}^{(i)})$ and the residual logical failure indicator
\begin{equation}
    L^{(i)}
    =
    L(\mathbf{m}^{(i)},\mathbf{y}^{(i)},\decoder).
\end{equation}
For shot $i$, let $\mathbf{E}^{(i)}=(E_\alpha^{(i)})_{\alpha\in\mathcal N}$ collect the events sampled across all noise components.
The entry $E_\alpha^{(i)}$ records which event occurs at component $\alpha$, including the no-error event, and may take more than two values when several error types are possible.
For example, for a single-qubit depolarizing channel with strength $\lambda_\alpha$, $E_\alpha^{(i)}$ is $I$, $X$, $Y$, or $Z$: $I$ occurs with probability $1-\lambda_\alpha$, while each nonidentity Pauli occurs with probability $\lambda_\alpha/3$.
Let $\pi_\alpha(e_\alpha;\lambda_\alpha) =\Prob_{\lambda_\alpha}(E_\alpha=e_\alpha)$ denote the probability of outcome $e_\alpha$ at component $\alpha$.
Under the independence assumption, the probability of sampling an error configuration $\mathbf e$ factorizes as
\begin{equation}
    \pi_{\bm{\lambda}}(\mathbf{e})
    =
    \Prob_{\bm{\lambda}}(\mathbf{E}=\mathbf{e})
    =
    \prod_{\beta\in\mathcal N}
    \pi_\beta(e_\beta;\lambda_\beta).
    \label{eq:error_configuration_probability}
\end{equation}
Conditioning on the sampled error configuration, the sensitivity can then be written as
\begin{align}
    \SL_\alpha
    &=
    \frac{\partial}{\partial\lambda_\alpha}
    \sum_{\mathbf e}
    \pi_{\bm{\lambda}}(\mathbf e)
    \E\!\left[
        L^{(i)}\mid\mathbf E^{(i)}=\mathbf e
    \right]
    \nonumber\\
    &=
    \sum_{\mathbf e}
    \pi_{\bm{\lambda}}(\mathbf e)
    \E\!\left[
        L^{(i)}\mid\mathbf E^{(i)}=\mathbf e
    \right]
    \frac{\partial}{\partial\lambda_\alpha}
    \log\pi_{\bm{\lambda}}(\mathbf e)
    \nonumber\\
    &=
    \sum_{\mathbf e}
    \pi_{\bm{\lambda}}(\mathbf e)
    \E\!\left[
        L^{(i)}\mid\mathbf E^{(i)}=\mathbf e
    \right]
    \frac{\partial}{\partial\lambda_\alpha}
    \log\pi_\alpha(e_\alpha;\lambda_\alpha).
    \label{eq:likelihood_ratio_identity}
\end{align}
We denote the final logarithmic derivative in Eq.~\eqref{eq:likelihood_ratio_identity} by $s_\alpha(e_\alpha)$ and refer to it as the score of component $\alpha$:
\begin{equation}
    s_\alpha(e_\alpha)
    =
    \frac{\partial}{\partial\lambda_\alpha}
    \log\pi_\alpha(e_\alpha;\lambda_\alpha),
    \qquad
    s_\alpha^{(i)}=s_\alpha(E_\alpha^{(i)}).
    \label{eq:local_channel_score}
\end{equation}
With this notation, Eq.~\eqref{eq:likelihood_ratio_identity} becomes
\begin{equation}
    \SL_\alpha
    =
    \E_{\bm{\lambda}}\!\left[
        L^{(i)}s_\alpha^{(i)}
    \right].
    \label{eq:score_sensitivity_identity}
\end{equation}

A standard logical error rate simulation samples $(\mathbf m^{(i)},\mathbf y^{(i)})$ and applies $\decoder$ to determine $L^{(i)}$ for each shot.
Averaging these failure indicators gives
\begin{equation}
    \widehat{P}_{\mathrm{L},N}
    =
    \frac{1}{N}\sum_{i=1}^{N}L^{(i)}.
    \label{eq:logical_failure_mc_estimator}
\end{equation}
To estimate the sensitivities in the same run, we first tabulate $s_\alpha(e_\alpha)$ for every possible outcome of each noise component.
During shot $i$, the sampled event $E_\alpha^{(i)}$ selects the corresponding score $s_\alpha^{(i)}$.
Once $L^{(i)}$ is known, we add $L^{(i)}s_\alpha^{(i)}$ to the accumulator for component $\alpha$.
After $N$ shots, the estimates are
\begin{equation}
    \widehat{\SL}_{\alpha,N}
    =
    \frac{1}{N}
    \sum_{i=1}^{N}
    L^{(i)}s_{\alpha}^{(i)},
    \qquad
    \widehat{\RL}_{\alpha,N}
    =
    \lambda_{\alpha}\widehat{\SL}_{\alpha,N},
    \label{eq:mc_estimator}
\end{equation}
where unbiasedness follows directly from Eq.~\eqref{eq:score_sensitivity_identity}:
\begin{proposition}[Unbiasedness]
\label{prop:score_estimator_unbiasedness}
For a fixed deterministic decoder, or a randomized decoder whose internal randomness is independent of $\bm{\lambda}$, the estimators in Eq.~\eqref{eq:mc_estimator} satisfy
\begin{equation}
        \E_{\bm{\lambda}}[\widehat{\SL}_{\alpha,N}]=\SL_\alpha,\quad
        \E_{\bm{\lambda}}[\widehat{\RL}_{\alpha,N}]=\RL_\alpha.
    \label{eq:unbiased_component_estimators}
\end{equation}
\end{proposition}
Hence, without additional circuit samples or decoder calls, the same simulation yields unbiased sensitivity and marginal-risk estimators by evaluating and accumulating the component scores.

For a physical component $c$ comprising elementary components $\alpha\in\mathcal N_c$ with a common error probability $p_c$, the corresponding estimates are
\begin{equation}
    \widehat{\SL}_{c,N}
    =
    \sum_{\alpha\in\mathcal N_c}
    \widehat{\SL}_{\alpha,N},
    \qquad
    \widehat{\RL}_{c,N}
    =
    p_c\widehat{\SL}_{c,N}.
    \label{eq:physical_component_estimator}
\end{equation}
By linearity and Proposition~\ref{prop:score_estimator_unbiasedness}, these are unbiased estimators of $\SL_c$ and $\RL_c$.
The same aggregation applies to gate types, syndrome rounds, qubits, boundaries, and spacetime regions.


\section{Applying error attribution to rotated surface code memories}
\label{sec:surface_code_results}

\begin{figure*}[!t]
    \centering
    \includegraphics[width=0.98\textwidth]{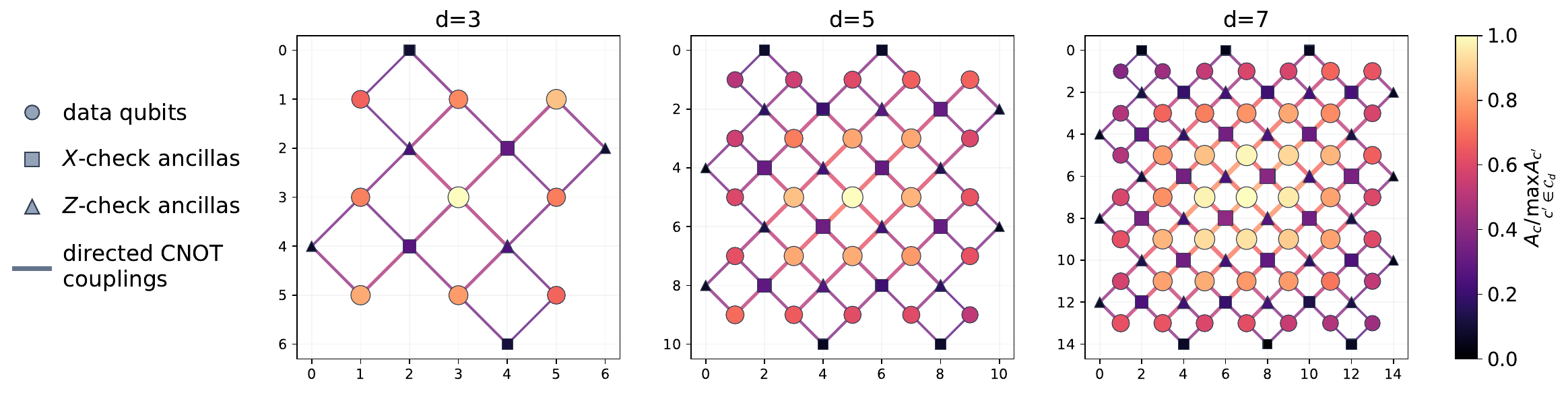}
    \caption{
    Physical component hotspot maps for rotated surface code memories at $p=0.001$ and $d$ syndrome rounds.  
    The color of each component encodes its within-distance relative sensitivity $\widetilde A_c^{(d)}=A_c/\max_{c'\in\mathcal C_d}A_{c'}$, where $A_c=\left|(S_c^{X}+S_c^{Z})/2\right|$ is the sensitivity magnitudes averaged over the $X$- and $Z$-memory, and $\mathcal C_d$ is the component set of the distance-$d$ circuit.
    The color bar is therefore shared across all panels and spans $[0,1]$.
    Circles denote data qubits, squares $X$-check ancillas, triangles $Z$-check ancillas, and links denote directed CNOT couplings.
    }
    \label{fig:surface_code_spatial_sensitivity}
\end{figure*}

We now concretely demonstrate the application of error attribution in the rotated surface code setting~\cite{fowler2012surface,gidney2021stim}.
For each $d\in\{3,5,7\}$, we simulate a distance-$d$ rotated surface code memory for $d$ rounds of syndrome extraction and using a uniform circuit-level noise parameter $p$, as in standard surface code benchmarks~\cite{gidney2021stim,higgott2025sparse}.
Each circuit contains $d^2$ data qubits, $d^2-1$ check ancillas, and $4d(d-1)$ distinct directed CNOT couplings.
We count each qubit or coupling as a single physical component across all syndrome rounds, giving $41$, $129$, and $265$ components for $d=3,5,7$, respectively.
Appendix~\ref{app:surface_code_protocol} gives the complete circuit, noise, decoder, and statistical protocol.
In the analysis below, we focus on the case $p=0.001$.

First observe that increasing the code distance from $d=3$ to $d=7$ reduces the logical error rate, although the number of noisy syndrome extraction rounds increases with $d$.
For the complete $d$-round experiments, the logical error rates averaged over the $X$- and $Z$-memory experiments are
\begin{small}
\begin{equation}
\begin{aligned}
    \overline{P}_{\mathrm{L}}^{(3)}
    &=8.17725\times10^{-4},\\
    \overline{P}_{\mathrm{L}}^{(5)}
    &=1.39429\times10^{-4},\\
    \overline{P}_{\mathrm{L}}^{(7)}
    &=1.89193\times10^{-5}.
\end{aligned}
\label{eq:surface_code_logical_rates}
\end{equation}
\end{small}

The corresponding sensitivities of the physical components are shown in Fig.~\ref{fig:surface_code_spatial_sensitivity}.
For all three code distances, the maps exhibit pronounced hotspots, with the largest sensitivities concentrated on a small subset of physical components.
At $d=3$, the three components with the highest sensitivities are the central data qubit component at coordinate $(3,3)$ and the two opposite corner data qubit components at $(5,1)$ and $(1,5)$, with the central component ranked first.
These three locations lie along the diagonal running from the lower left to the upper right (the ``/'' diagonal) in the displayed patch.
At $d=5$ and $d=7$, components with large sensitivities are likewise predominantly data qubit components, concentrated near the center of the patch and tend to lie along the ``/'' diagonal.
The spatial distributions of the seven and thirteen highest-sensitivity components for $d=5$ and $d=7$, respectively, are shown in Fig.~\ref{fig:hotspot_targeted_intervention}(a).
The features of this spatial pattern can be understood intuitively.
In an ideal model, the boundaries limit the ways a local error can combine with other errors to cause a logical failure after decoding.
Errors near the center can participate in more such combinations, making them more likely to cause a logical failure.
The diagonal tendency follows from the CNOT ordering: data qubits along the ``/'' direction may accumulate more propagated errors, making local errors at these locations more likely to cause logical failure.
Appendix~\ref{app:hotspot_geometry} develops these arguments in detail.

The leading components in Fig.~\ref{fig:hotspot_targeted_intervention}(a) carry more linearized marginal risk than their number alone would suggest.
Specifically, at $d=3,5,7$, the leading sets contain $7.32\%$, $5.43\%$, and $4.91\%$ of all physical components, respectively, but account for $15.1\%$, $10.2\%$, and $8.9\%$ of the total linearized marginal risk.
We also report the distribution of total linearized marginal risk across component types.
At $d=3,5,7$, data qubit components constitute $22.0\%$, $19.4\%$, and $18.5\%$ of all components, but account for $39.5\%$, $29.5\%$, and $26.1\%$ of the total linearized marginal risk, respectively.
At the same distances, ancilla-qubit components constitute $19.5\%$, $18.6\%$, and $18.1\%$ of all components, but account for only $7.9\%$, $7.5\%$, and $6.9\%$ of the risk.
CNOT components constitute the remaining $58.5\%$, $62.0\%$, and $63.4\%$ of all components and account for $52.6\%$, $63.0\%$, and $67.0\%$ of the risk.
These results show that errors associated with different circuit components have markedly unequal effects on logical performance.
This nonuniformity motivates the targeted intervention, which we examine in the next section.

\section{Hotspot-guided targeted circuit intervention}
\label{sec:hotspot_intervention}

Hotspot information suggests a quantum circuit analogue of targeted therapy: when only a limited number of physical components can be improved, the intervention can be aimed at those with the largest sensitivity instead of spread uniformly over the circuit. We demonstrate the practical benefits of this strategy in both rotated surface code memories and high-rate lifted product code instances of intensive interest in practical QEC. 

\subsection{Rotated surface codes}
\label{subsec:intervention_surface}

\begin{figure*}[!t]
    \centering
    \includegraphics[width=0.98\textwidth]{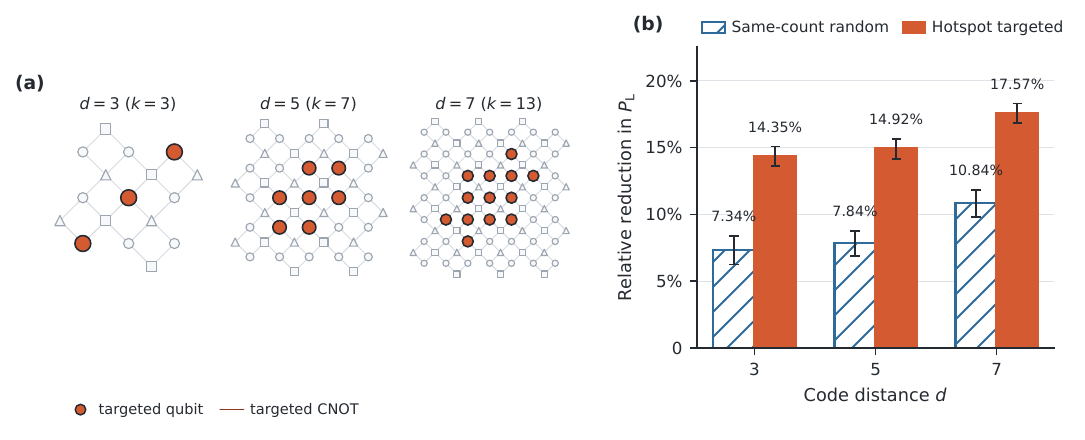}
    \caption{
    Hotspot-guided targeted intervention at $p=0.001$.
    (a) The $k$ components with the highest sensitivities on physical layouts.
    Gray marks and links show the qubits and nearest-neighbor coupling graph; marker shapes follow Fig.~\ref{fig:surface_code_spatial_sensitivity}.
    Orange filled marks denote the components with the highest sensitivities.
    The sets contain $k=3$ ($7.32\%$), $7$ ($5.43\%$), and $13$ ($4.91\%$) components at $d=3$, $5$, and $7$, respectively.
    (b) Relative logical error rate reduction with respect to the circuit without intervention.
    Orange bars show the reduction from hotspot-guided intervention, while blue hatched bars show the mean reduction from random intervention over 40 sets of $k$ components chosen uniformly at random.
    The local error rate of each selected component is reduced by a fraction $\alpha=0.5$.
    Error bars indicate 95\% bootstrap confidence intervals.
    }
    \label{fig:hotspot_targeted_intervention}
\end{figure*}

We first apply this strategy to the rotated surface code memories studied in Sec.~\ref{sec:surface_code_results}.
Based on the attribution results at $p=0.001$, we select the highest-sensitivity components as intervention targets: $3$ of the $41$ total components ($7.32\%$), $7$ of the $129$ ($5.43\%$), and $13$ of the $265$ ($4.91\%$) at $d=3,5,7$, respectively, as shown in Fig.~\ref{fig:hotspot_targeted_intervention}(a).
For each selected component $\ell$, we reduce its local error rate by a fraction $\alpha$:
\begin{equation}
    p'_{\ell}=(1-\alpha)p_{\ell},
    \qquad \alpha=0.5,
    \label{eq:targeted_intervention}
\end{equation}
The rescaling above is an idealized proxy for component-selective improvement.
Experimentally, component-selective improvements could be pursued through calibration and control tailored to the selected component.
For superconducting qubits, tuning pulse amplitudes, detunings, and shapes can reduce single-qubit gate errors and leakage~\cite{chen2016measuring}, while dynamical-decoupling sequences can improve state preservation during idle periods~\cite{pokharel2018demonstration}.
For two-qubit operations, pulse optimization in tunable-coupler architectures has enabled high-fidelity CZ and $ZZ$-free iSWAP gates~\cite{sung2021realization}.
At the processor level, qubit-frequency trajectories can be optimized jointly, with selective reoptimization of poorly performing gates~\cite{klimov2024optimizing}.

After hotspot-guided intervention, the logical error rates averaged over the $X$- and $Z$-memory experiments are
\begin{small}
\begin{equation}
\begin{aligned}
    \overline{P}_{\mathrm{L},\mathrm{t}}^{(3)}
    &=7.03120\times10^{-4},\\
    \overline{P}_{\mathrm{L},\mathrm{t}}^{(5)}
    &=1.18580\times10^{-4},\\
    \overline{P}_{\mathrm{L},\mathrm{t}}^{(7)}
    &=1.56532\times10^{-5}.
\end{aligned}
\label{eq:surface_code_targeted_logical_rates}
\end{equation}
\end{small}

We compare the hotspot-guided intervention against two references: the circuit without intervention and an ensemble of circuits in which the same intervention is applied to the same number of physical components selected uniformly at random.
The first reference serves as the baseline, while the second tests whether hotspot information improves target selection under a fixed intervention budget.

Fig.~\ref{fig:hotspot_targeted_intervention} shows the comparison results.
Relative to the circuit without intervention, the hotspot-guided intervention reduces the logical error rates by $14.35\%$, $14.92\%$, and $17.57\%$ for $d=3,5,7$, respectively.
For comparison, applying the same intervention to the same number of physical components selected uniformly at random reduces the logical error rates by $7.34\%$, $7.84\%$, and $10.84\%$ on average for $d=3,5,7$, respectively.
The corresponding improvement ratio is $\mathcal I=\Delta_{\mathrm{t}}/\Delta_{\mathrm{r}}$, where $\Delta_{\mathrm{t}}$ and $\Delta_{\mathrm{r}}$ denote the targeted and mean random reductions, respectively. 
The estimated ratios are $1.95$ [95\% confidence interval (CI), $1.72$--$2.26$] at $d=3$, $1.90$ [$1.71$--$2.14$] at $d=5$, and $1.62$ [$1.49$--$1.77$] at $d=7$.
All three confidence intervals lie entirely above one, indicating a benefit of hotspot-informed selection over random selection.
In these circuits, halving the local error rates of only $4.9\%$--$7.3\%$ of the physical components reduces logical error rates by $14.3\%$--$17.6\%$.
The complete selection, sampling, and interval constructions are given in Appendix~\ref{app:hotspot_intervention}.

\begin{figure}[!t]
    \centering
    \includegraphics[width=\columnwidth]{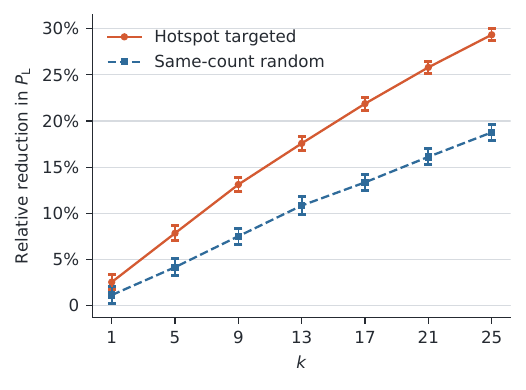}
    \caption{
    Relative logical error rate reduction versus intervention budget.
    The $d=7$ rotated surface code circuit contains 265 physical components and uses seven syndrome rounds at $p=0.001$.
    Orange circles show the reduction from hotspot-guided intervention, while blue squares show the mean reduction from random intervention over 40 sets of $k$ components chosen uniformly at random.
    The local error rate of each selected component is reduced by a fraction $\alpha=0.5$.
    Error bars indicate 95\% bootstrap confidence intervals.
    }
    \label{fig:surface_code_budget_sweep}
\end{figure}

We also examine how the benefit varies with the intervention budget.
Fig.~\ref{fig:surface_code_budget_sweep} compares interventions on the most sensitive components with those on randomly selected components across different intervention budgets for the $d=7$ circuit.
As the number of selected components increases from $1$ to $25$, the estimated reduction rises from $2.55\%$ to $29.34\%$ under hotspot-guided intervention.
For randomly selected components, the mean reduction rises from $1.15\%$ to $18.76\%$ and remains below the targeted estimate at every budget.
These results demonstrate how sensitivity information helps achieve significantly more efficient interventions across the budgets.

\subsection{High-rate lifted product codes}
\label{subsec:intervention_high_rate}

To test hotspot-guided intervention beyond low-rate, geometric codes, we study the instances in the recent mitten code family~\cite{mitten2026processors}:
$[\![150,30,10]\!]$ over $C_5\times S_3$, $[\![200,40,12]\!]$ over $C_4\times D_{10}$, and $[\![300,60,14]\!]$ over $C_{10}\times S_3$.
These high-rate non-Abelian lifted product qLDPC codes have an encoding rate $1/5$ and weight-9 stabilizer checks with geometrically nonlocal connectivity.
We reconstruct the syndrome extraction circuits from the published check matrices, following the circuit structure described in Ref.~\cite{mitten2026processors}.
We simulate the three codes for $10$, $12$, and $14$ syndrome rounds, respectively.
At $p=0.002$, the logical error rates averaged over the $X$- and $Z$-memory experiments are
\begin{small}
\begin{equation}
\begin{aligned}
    \overline{P}_{\mathrm{L}}^{[\![150,30,10]\!]}
    &=5.48966\times10^{-5},\\
    \overline{P}_{\mathrm{L}}^{[\![200,40,12]\!]}
    &=4.52000\times10^{-6},\\
    \overline{P}_{\mathrm{L}}^{[\![300,60,14]\!]}
    &=2.27493\times10^{-6}.
\end{aligned}
\label{eq:high_rate_logical_rates}
\end{equation}
\end{small}
Appendix~\ref{app:high_rate_intervention} gives the circuit construction, complete experimental protocol, and detailed results.

The syndrome extraction circuits for the three codes contain $N=1350$, $1800$, and $2700$ physical components, respectively.
These comprise $150$, $200$, and $300$ data qubits; $120$, $160$, and $240$ check ancillas, split equally between $X$ and $Z$ checks; and $1080$, $1440$, and $2160$ directed CNOT couplings.
We use the attribution framework to select components with large sensitivity magnitudes for intervention.
The selected sets contain $84$, $112$, and $168$ components, respectively, each accounting for $6.22\%$ of all components in its circuit.
Of the selected components, $77$ ($91.67\%$), $85$ ($75.89\%$), and $132$ ($78.57\%$) are directed CNOT couplings, respectively.
The remaining $7$ ($8.33\%$), $27$ ($24.11\%$), and $36$ ($21.43\%$) components are check ancillas; no data qubits are selected.
For each selected component, we reduce its local error rate by a fraction $\alpha=0.5$, as in Eq.~\eqref{eq:targeted_intervention}.
Similarly, we compare the hotspot-guided intervention with the baseline without intervention and with the same intervention on an equal number of components selected randomly.

\begin{figure}[!tp]
    \centering
    \includegraphics[width=\columnwidth]{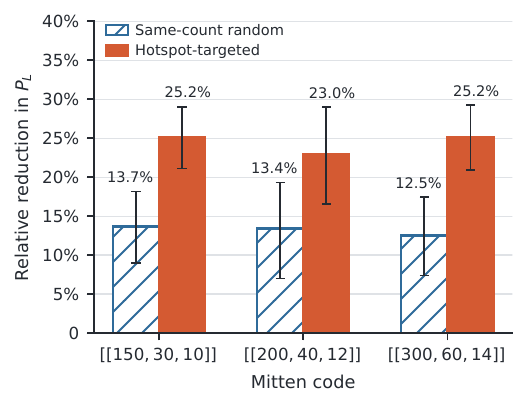}
    \caption{
    Relative logical error rate reduction for three high-rate lifted product codes at $p=0.002$.
    Each targeted set contains $6.22\%$ of the physical components in its circuit: $k=84$, $112$, and $168$ for the Mitten $[\![150,30,10]\!]$, $[\![200,40,12]\!]$, and $[\![300,60,14]\!]$ codes, respectively.
    Orange bars show the reduction from hotspot-guided intervention, while blue hatched bars show the mean reduction from random intervention over 40 sets of $k$ components chosen uniformly at random.
    The local error rate of each selected component is reduced by a fraction $\alpha=0.5$.
    Error bars indicate 95\% bootstrap confidence intervals.
    }
    \label{fig:high_rate_intervention}
\end{figure}

Fig.~\ref{fig:high_rate_intervention} compares the targeted and random interventions in all three codes.
For $[\![150,30,10]\!]$, the reduction relative to the circuit without intervention is $25.16\%$, compared with $13.69\%$ for uniform selection; for $[\![200,40,12]\!]$, the corresponding reductions are $23.01\%$ and $13.42\%$.
For $[\![300,60,14]\!]$, they are $25.17\%$ and $12.51\%$.
The corresponding improvement ratios, $\mathcal I=\Delta_{\mathrm{t}}/\Delta_{\mathrm{r}}$, are $1.84$ [95\% CI, $1.41$--$2.65$], $1.71$ [$1.20$--$2.95$], and $2.01$ [$1.44$--$3.29$], respectively.

These results further demonstrate the value of hotspot-guided intervention for high-rate qLDPC codes underpinning recent proposals for low-overhead fault-tolerant processors~\cite{mitten2026processors}, with potential to guide ongoing efforts toward their implementation on e.g.~reconfigurable atom array platforms.


\section{Discussion}

We introduced the error attribution method that leverages classical simulation techniques to identify noise bottlenecks in QEC procedures and guide targeted improvements of their performance.
Our estimator obtains all component sensitivities from the same Monte Carlo samples used to estimate the logical error rate, without additional decoder evaluations.
Across the two structurally distinct code families studied here, this information improves target selection under a limited component budget.
The differing hotspot patterns highlight the importance of assessing intervention priorities for a given circuit, noise model, and decoder, even when the elementary error probabilities are uniform.

Future work could extend the framework to account for both achievable noise reductions and calibration costs when prioritizing components for intervention.
Noise models inferred from device measurements~\cite{acharya2023suppressing,acharya2025belowthreshold} could be combined with estimates of how available calibration controls change component error probabilities.
Component sensitivities could then be used to estimate the logical benefit of feasible control changes and compare that benefit with the resources required.
Further studies could also examine how spatial and temporal noise correlations affect intervention priorities.
Testing and refining the framework on experimental devices could guide more effective targeted calibration and accelerate progress toward practical quantum computing.

\begin{acknowledgments}
We thank Zhenyu Chen, Song Cheng, Zhengwei Liu and Zhaohui Wei for valuable discussions.
YS is supported by a startup funding from YMSC, Tsinghua Shuimu Scholar Program, and China Postdoctoral Science Foundation (Grant No.~2026M793709).
ZWL is supported in part by National Natural Science Foundation of China (Grant No.~12475023), Dushi Program, and a startup funding from the YMSC.
\end{acknowledgments}

\section*{Data Availability}

The data supporting this study are available from the authors upon reasonable request.
The simulations were performed using the FaultScope software package~\cite{faultscope}.

\appendix

\section{Rotated surface code benchmark}
\label{app:surface_code_protocol}

\subsection{Code construction and syndrome extraction}

We consider the rotated surface code, which encodes one logical qubit in a square array of $d^2$ data qubits using local $X$- and $Z$-parity checks~\cite{fowler2012surface,tomita2014low}.
Our benchmarks use the canonical memory circuits at distances $d\in\{3,5,7\}$, with $d$ rounds of syndrome extraction~\cite{gidney2021stim}.

\paragraph{Code layout and logical operators.}
The data qubits are located at the vertices of a $d\times d$ square grid.
Each square has an $X$- or $Z$-type check acting on the four data qubits at its corners.
Squares sharing an edge have opposite check types.
Additional weight-two checks are placed along the boundary; the weight of a check is the number of data qubits on which it acts.
The weight-two $X$ checks lie along the top and bottom edges, and the weight-two $Z$ checks lie along the left and right edges.
Each stabilizer is measured using a separate ancilla qubit.
For the odd distances considered here, each check type has $(d-1)^2/2$ weight-four bulk checks and $d-1$ weight-two boundary checks, giving $d^2-1$ check ancillas in total.

We place the data qubits at coordinates $(2i-1,2j-1)$, with $i,j=1,\ldots,d$.
Each ancilla has even integer coordinates and is coupled to four neighboring data qubits for a bulk check or two for a boundary check.
Let $\mathbf r_a$ denote the coordinate of ancilla $a$.
Write $q_{\delta_x,\delta_y}(a)$ for the neighboring data qubit at offset $(\delta_x,\delta_y)$ from $a$, there is
\begin{equation}
    \mathbf r\!\left(q_{\delta_x,\delta_y}(a)\right)
    =\mathbf r_a+(\delta_x,\delta_y),
    \qquad \delta_x,\delta_y\in\{+1,-1\}.
    \label{eq:surface_code_check_coordinates}
\end{equation}
Let $N(a)$ be the set of data qubits coupled to ancilla $a$.
For an $X$- or $Z$-check ancilla $a$, the measured stabilizer is, respectively,
\begin{equation}
    K_a^X=\prod_{q\in N(a)}X_q,
    \qquad
    K_a^Z=\prod_{q\in N(a)}Z_q.
    \label{eq:surface_code_check_operators}
\end{equation}
A bulk check acts on all four corners of its cell; a boundary check acts on two adjacent data qubits.
The two types of stabilizer commute because an $X$ check and a $Z$ check share either zero or two data qubits.
The code space is their simultaneous $+1$ eigenspace.
There are $d^2-1$ independent stabilizer generators on $d^2$ data qubits, leaving a two-dimensional space that encodes one logical qubit.
A data qubit $X$ error flips the eigenvalues of the adjacent $Z$ checks, and a $Z$ error flips those of the adjacent $X$ checks; a $Y$ error can flip both types.
These joint-parity measurements reveal error information without measuring the encoded logical state~\cite{fowler2012surface}.

Logical operators commute with every stabilizer but act nontrivially on the encoded qubit.
In this layout, the logical Pauli operators are
\begin{equation}
    \overline X=\prod_{q\in\gamma_X}X_q,
    \qquad
    \overline Z=\prod_{q\in\gamma_Z}Z_q,
    \label{eq:surface_code_logical_operators}
\end{equation}
where $\gamma_X$ is the leftmost column of data qubits and $\gamma_Z$ is the top row.
These strings connect opposite edges of the patch and intersect at one data qubit, so $\overline X$ and $\overline Z$ anticommute.
The shortest nontrivial logical operator has weight $d$, giving the code parameters $[\![d^2,1,d]\!]$~\cite{tomita2014low,orourke2025compare}.

\paragraph{Syndrome extraction.}
In each round of the ideal circuit, both $X$- and $Z$-type stabilizers are measured using ancillas initialized in $\lvert0\rangle$, which interact with neighboring data qubits through CNOT gates and are then measured and reset to $\lvert0\rangle$ for the next round.
The CNOT gates are scheduled in four layers, with gates in each layer acting in parallel and no two gates in the same layer sharing a qubit.

For an $X$ check, a Hadamard $H_a$ first prepares the ancilla in $\lvert+\rangle$.
The ancilla controls a CNOT onto each data qubit in the check, and a second $H_a$ followed by $Z$-basis readout measures the ancilla in the $X$ basis.
The resulting measurement bit $m_a$ corresponds to the stabilizer eigenvalue $(-1)^{m_a}$~\cite{fowler2012surface,tomita2014low}.
For a bulk $X$ check, the CNOT gates in layers $1,2,3,4$ target the data qubits at the following coordinate offsets from ancilla $a$, respectively:
\begin{equation}
    \mathcal O_X
    =\bigl((+1,+1),(-1,+1),(+1,-1),(-1,-1)\bigr).
    \label{eq:surface_code_x_check_order}
\end{equation}
For an $X$ check on the top boundary, the ancilla controls two CNOTs in layers $1$ and $2$, targeting the two data qubits below it.
For an $X$ check on the bottom boundary, the two CNOTs occur in layers $3$ and $4$.
Fig.~\ref{fig:surface_code_x_check_circuit} illustrates this measurement circuit.

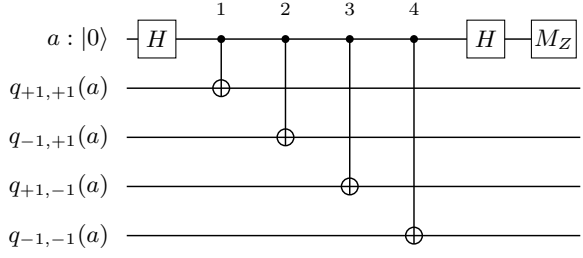
\begin{figure}[t]
    \centering
    \begin{tikzpicture}[
        x=1cm, y=1cm,
        wire/.style={line width=0.55pt},
        gate/.style={draw,fill=white,minimum width=5mm,minimum height=5mm,
                     inner sep=1pt,font=\small},
        wirelabel/.style={anchor=east,font=\small},
        layerlabel/.style={font=\scriptsize}
    ]
        \draw[wire] (0,0) -- (5.65,0);
        \foreach \y in {-0.65,-1.30,-1.95,-2.60}
            \draw[wire] (0,\y) -- (6.0,\y);
        \node[wirelabel] at (-0.12,0) {$a:\lvert0\rangle$};
        \node[wirelabel] at (-0.12,-0.65) {$q_{+1,+1}(a)$};
        \node[wirelabel] at (-0.12,-1.30) {$q_{-1,+1}(a)$};
        \node[wirelabel] at (-0.12,-1.95) {$q_{+1,-1}(a)$};
        \node[wirelabel] at (-0.12,-2.60) {$q_{-1,-1}(a)$};
        \node[gate] at (0.40,0) {$H$};
        \node[gate] at (4.75,0) {$H$};
        \node[gate] at (5.65,0) {$M_Z$};

        \foreach \x/\y/\j in {1.25/-0.65/1,2.10/-1.30/2,
                                2.95/-1.95/3,3.80/-2.60/4} {
            \draw[wire] (\x,0) -- (\x,\y);
            \fill (\x,0) circle (1.5pt);
            \draw[wire,fill=white] (\x,\y) circle (0.11);
            \draw[wire] (\x-0.11,\y) -- (\x+0.11,\y);
            \draw[wire] (\x,\y-0.11) -- (\x,\y+0.11);
            \node[layerlabel] at (\x,0.40) {$\j$};
        }
    \end{tikzpicture}
    \caption{
    Measurement circuit for a bulk $X$ stabilizer in the rotated surface code.
    The ancilla $a$, initialized in $\lvert0\rangle$, controls four CNOTs onto neighboring data qubits.
    Numbers indicate the CNOT layers, following the offset order $\mathcal O_X$ in Eq.~\eqref{eq:surface_code_x_check_order}.
    $H$ denotes a Hadamard gate, and $M_Z$ denotes measurement in the $Z$ basis.
    An $X$ check on the top (bottom) boundary uses only layers $1,2$ ($3,4$) and the corresponding two data qubits, with the same ancilla preparation, Hadamard gates, and measurement.
    }
    \label{fig:surface_code_x_check_circuit}
\end{figure}

For a $Z$ check, each data qubit is instead the CNOT control and $a$ is the target.
The ancilla is then measured directly in the $Z$ basis, so no Hadamard gates are applied to the ancilla qubit.
For a bulk $Z$ check, the CNOT gates in layers $1,2,3,4$ are controlled by the data qubits at the following coordinates, respectively:
\begin{equation}
    \mathcal O_Z
    =\bigl((+1,+1),(+1,-1),(-1,+1),(-1,-1)\bigr).
    \label{eq:surface_code_z_check_order}
\end{equation}
For a $Z$ check on the left boundary, the ancilla is the target of two CNOTs in layers $1$ and $2$, controlled by the two data qubits to its right.
For a $Z$ check on the right boundary, the two CNOTs occur in layers $3$ and $4$.
Fig.~\ref{fig:surface_code_check_circuits} shows the CNOT ordering for $X$ and $Z$ checks in the bulk and at the boundaries.

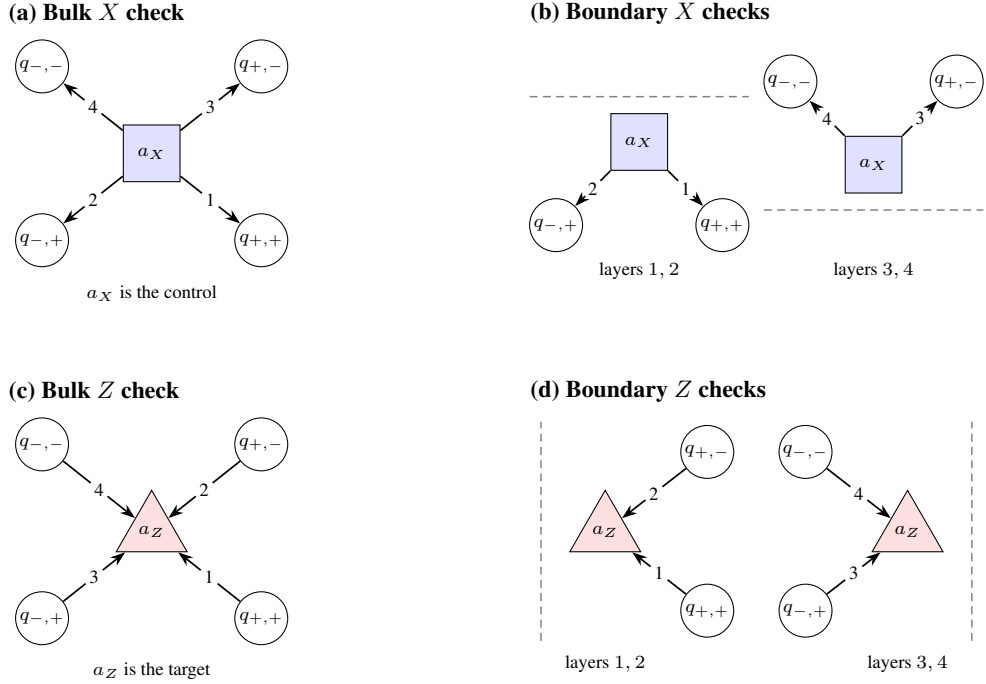
\begin{figure*}[t]
    \centering
    \begin{tikzpicture}[
        x=1cm,
        y=1cm,
        data/.style={circle,draw=black,fill=white,minimum size=7mm,
                     inner sep=0pt,font=\scriptsize},
        xcheck/.style={rectangle,draw=black,fill=blue!12,minimum size=7.5mm,
                       inner sep=0pt,font=\scriptsize},
        zcheck/.style={regular polygon,regular polygon sides=3,draw=black,
                       fill=red!12,minimum size=10mm,inner sep=0.5pt,
                       font=\scriptsize},
        cnot/.style={-{Stealth[length=2mm,width=1.4mm]},line width=0.65pt},
        layer/.style={circle,fill=white,inner sep=0.8pt,font=\scriptsize},
        panel/.style={font=\small\bfseries,anchor=west},
        note/.style={font=\scriptsize,align=center},
        boundary/.style={gray,densely dashed,line width=0.55pt}
    ]
        \begin{scope}[shift={(2.1,0)}]
            \node[panel] at (-2.0,1.85) {(a) Bulk $X$ check};
            \node[xcheck] (xa) at (0,0) {$a_X$};
            \node[data] (xpp) at (1.45,-1.15) {$q_{+,+}$};
            \node[data] (xmp) at (-1.45,-1.15) {$q_{-,+}$};
            \node[data] (xpm) at (1.45,1.15) {$q_{+,-}$};
            \node[data] (xmm) at (-1.45,1.15) {$q_{-,-}$};
            \draw[cnot] (xa) -- node[layer] {1} (xpp);
            \draw[cnot] (xa) -- node[layer] {2} (xmp);
            \draw[cnot] (xa) -- node[layer] {3} (xpm);
            \draw[cnot] (xa) -- node[layer] {4} (xmm);
            \node[note] at (0,-1.85) {$a_X$ is the control};
        \end{scope}

        \begin{scope}[shift={(10.1,0)}]
            \node[panel] at (-3.1,1.85) {(b) Boundary $X$ checks};
            \node[xcheck] (xb12) at (-1.55,0.15) {$a_X$};
            \node[data] (xbpp) at (-0.45,-0.95) {$q_{+,+}$};
            \node[data] (xbmp) at (-2.65,-0.95) {$q_{-,+}$};
            \draw[cnot] (xb12) -- node[layer] {1} (xbpp);
            \draw[cnot] (xb12) -- node[layer] {2} (xbmp);
            \draw[boundary] (-3.0,0.75) -- (-0.1,0.75);
            \node[note] at (-1.55,-1.55) {layers $1,2$};

            \node[xcheck] (xb34) at (1.55,-0.15) {$a_X$};
            \node[data] (xbpm) at (2.65,0.95) {$q_{+,-}$};
            \node[data] (xbmm) at (0.45,0.95) {$q_{-,-}$};
            \draw[cnot] (xb34) -- node[layer] {3} (xbpm);
            \draw[cnot] (xb34) -- node[layer] {4} (xbmm);
            \draw[boundary] (0.1,-0.75) -- (3.0,-0.75);
            \node[note] at (1.55,-1.55) {layers $3,4$};
        \end{scope}

        \begin{scope}[shift={(2.1,-5.0)}]
            \node[panel] at (-2.0,1.85) {(c) Bulk $Z$ check};
            \node[zcheck] (za) at (0,0) {$a_Z$};
            \node[data] (zpp) at (1.45,-1.15) {$q_{+,+}$};
            \node[data] (zmp) at (-1.45,-1.15) {$q_{-,+}$};
            \node[data] (zpm) at (1.45,1.15) {$q_{+,-}$};
            \node[data] (zmm) at (-1.45,1.15) {$q_{-,-}$};
            \draw[cnot] (zpp) -- node[layer] {1} (za);
            \draw[cnot] (zpm) -- node[layer] {2} (za);
            \draw[cnot] (zmp) -- node[layer] {3} (za);
            \draw[cnot] (zmm) -- node[layer] {4} (za);
            \node[note] at (0,-1.85) {$a_Z$ is the target};
        \end{scope}

        \begin{scope}[shift={(10.1,-5.0)}]
            \node[panel] at (-3.1,1.85) {(d) Boundary $Z$ checks};
            \node[zcheck] (zb12) at (-2.0,0) {$a_Z$};
            \node[data] (zbpp) at (-0.65,-1.05) {$q_{+,+}$};
            \node[data] (zbpm) at (-0.65,1.05) {$q_{+,-}$};
            \draw[cnot] (zbpp) -- node[layer] {1} (zb12);
            \draw[cnot] (zbpm) -- node[layer] {2} (zb12);
            \draw[boundary] (-2.85,-1.45) -- (-2.85,1.45);
            \node[note] at (-2.0,-1.75) {layers $1,2$};

            \node[zcheck] (zb34) at (2.0,0) {$a_Z$};
            \node[data] (zbmp) at (0.65,-1.05) {$q_{-,+}$};
            \node[data] (zbmm) at (0.65,1.05) {$q_{-,-}$};
            \draw[cnot] (zbmp) -- node[layer] {3} (zb34);
            \draw[cnot] (zbmm) -- node[layer] {4} (zb34);
            \draw[boundary] (2.85,-1.45) -- (2.85,1.45);
            \node[note] at (2.0,-1.75) {layers $3,4$};
        \end{scope}
    \end{tikzpicture}
    \caption{
    CNOT ordering for rotated surface code stabilizer measurements.
    (a,b) Bulk and boundary $X$ checks.
    (c,d) Bulk and boundary $Z$ checks.
    Circles denote data qubits, squares $X$-check ancillas, and triangles $Z$-check ancillas.
    Numbers indicate the four CNOT layers.
    Bulk checks couple to four data qubits, while boundary checks couple to two.
    Dashed lines mark the patch boundaries.
    The top $X$ and left $Z$ boundary checks use layers $1,2$; the bottom $X$ and right $Z$ boundary checks use layers $3,4$.
    }
    \label{fig:surface_code_check_circuits}
\end{figure*}

\paragraph{Memory experiments.}
The $X$-memory experiment initializes every data qubit in $\lvert+\rangle$, performs $d$ rounds of syndrome extraction, and measures every data qubit in the $X$ basis at the end.
The $Z$-memory experiment uses the same extraction circuit, with data initialization in $\lvert0\rangle$ and final readout in the $Z$ basis.
In the absence of errors, the initial states are $+1$ eigenstates of $\overline X$ and $\overline Z$, respectively.
In the first round, $Z$ checks in the $X$ memory and $X$ checks in the $Z$ memory give random outcomes.
These measurements project the data into the corresponding stabilizer eigenspace, so subsequent rounds give the same outcomes in the absence of errors.

Let $m_a(t)\in\mathbb F_2$ be the measurement result of check $a$ in round $t$, with $0$ denoting the $+1$ eigenvalue and $1$ denoting the $-1$ eigenvalue.
Between consecutive rounds, a detection event occurs when the same check gives different measurement outcomes~\cite{orourke2025compare,gidney2021stim}.
For each check $a$ and $t\ge2$, we record detection events using the binary variable
\begin{equation}
    d_a(t)=m_a(t)\oplus m_a(t-1),
    \qquad t=2,\ldots,d,
    \label{eq:surface_code_round_detector}
\end{equation}
where $\oplus$ denotes addition modulo two, and a detection event occurs when $d_a(t)=1$.
For the ideal initial state, every $X$ check in the $X$ memory and every $Z$ check in the $Z$ memory has eigenvalue $+1$.
For these checks, the first measurement is compared with the expected bit $0$:
\begin{equation}
    d_a(1)=m_a(1).
    \label{eq:surface_code_initial_detector}
\end{equation}
The other checks have random first-round outcomes, so no such comparison is available.

The final data readout also determines the $X$-check outcomes in the $X$ memory and the $Z$-check outcomes in the $Z$ memory.
Let $b_q\in\mathbb F_2$ be the final measurement result of data qubit $q$.
We compare the XOR of the final measurement results on $N(a)$ with the outcome of check $a$ in the last round:
\begin{equation}
    d_a(d+1)=m_a(d)\oplus\bigoplus_{q\in N(a)}b_q.
    \label{eq:surface_code_boundary_detectors}
\end{equation}
For the odd distances, the $d-1$ pairs of consecutive rounds contribute $(d-1)(d^2-1)$ detectors; the initial and final comparisons contribute $(d^2-1)/2$ each, giving $d(d^2-1)$ detectors in total.

The measured logical result is the XOR of the final data-measurement bits along the leftmost column ($\gamma_X$) in the $X$ memory or the top row ($\gamma_Z$) in the $Z$ memory.
The decoder then uses the detector record to infer whether this result needs to be flipped.

\subsection{Noise model and decoding}

We use the standard circuit-level depolarizing noise model for surface code benchmarks~\cite{gidney2021stim,higgott2025sparse}.
Each one-qubit Clifford gate is followed by a one-qubit depolarizing channel with error probability $p$.
Each CNOT is followed by a two-qubit depolarizing channel acting jointly on its control and target qubits, also with total error probability $p$.
Independent one-qubit depolarizing channels with error probability $p$ are also applied to all data qubits before each round of syndrome extraction.
Measurement results are flipped with probability $p$.
A reset error occurs with probability $p$, preparing $\lvert1\rangle$ instead of $\lvert0\rangle$ or $\lvert-\rangle$ instead of $\lvert+\rangle$.
For a $k$-qubit depolarizing channel,
\begin{equation}
    \mathcal{E}_{p}(\rho)
    =
    (1-p)\rho
    +
    \frac{p}{4^k-1}
    \sum_{P\in\mathcal{P}_k\setminus\{I\}}P\rho P,
    \label{eq:depolarizing_channel}
\end{equation}
where $\mathcal{P}_k=\{I,X,Y,Z\}^{\otimes k}$ is the set of $k$-qubit Pauli operators.
For the two-qubit channel ($k=2$), $I\otimes I$ occurs with probability $1-p$, and each of the other $15$ Pauli products occurs with probability $p/15$.

We decode the detector record using minimum-weight perfect matching, implemented in PyMatching~\cite{higgott2025sparse}.
For each circuit, the decoding graph is derived from its circuit-level noise model using Stim~\cite{gidney2021stim}, with edge weights determined by the corresponding error probabilities.

\subsection{Simulation and statistical analysis}

We simulate the noisy circuits with the Forward simulator in FaultScope~\cite{faultscope}.
At $p=0.001$, we estimate the logical error rate and the sensitivities of all physical components from the same Monte Carlo samples.
For each memory basis, we use a total of $4\times10^7$, $4\times10^8$, and $4\times10^9$ shots at $d=3,5,7$, respectively.
Summed over both bases, the three distances therefore use $8\times10^7$, $8\times10^8$, and $8\times10^9$ shots and record 65,418, 111,543, and 151,354 logical failures.

Within each memory basis, the logical error rate is estimated as the total number of logical failures divided by the total number of shots.
We denote the resulting $X$- and $Z$-memory estimates by $P_{\mathrm{L}}^{X}$ and $P_{\mathrm{L}}^{Z}$, respectively.
The reported logical error rate is then obtained by averaging these two estimates:
\begin{equation}
    \overline{P}_{\mathrm{L}}
    =\frac{P_{\mathrm{L}}^{X}+P_{\mathrm{L}}^{Z}}{2}.
    \label{eq:surface_code_basis_averaged_logical_rate}
\end{equation}

We use bootstrap resampling to quantify the statistical uncertainty in the reported logical error rate.
Let $N_b$ denote the total number of shots for memory basis $b\in\{X,Z\}$.
In each bootstrap replicate, we independently generate a failure count $F_b^{*}$ for each basis by treating the $N_b$ shots as independent samples, each with failure probability $P_{\mathrm{L}}^{b}$:
\begin{equation}
    F_b^{*}\sim\operatorname{Binomial}(N_b,P_{\mathrm{L}}^{b}).
    \label{eq:surface_code_bootstrap_counts}
\end{equation}
We then average the two rates $F_X^{*}/N_X$ and $F_Z^{*}/N_Z$.
The 95\% confidence intervals in Table~\ref{tab:surface_code_p0p001} are given by the 2.5th and 97.5th percentiles of these averages over 50,000 bootstrap replicates.
Table~\ref{tab:surface_code_p0p001} summarizes the results at $p=0.001$.

\begin{table*}[t]
\caption{
Logical error rates for rotated surface code memories at $p=0.001$ with $d$ syndrome rounds.
Shot and failure counts are summed over the $X$- and $Z$-memory experiments.
$\overline{P}_{\mathrm{L}}$ is the average of the $X$- and $Z$-memory logical error rates, as defined in Eq.~\eqref{eq:surface_code_basis_averaged_logical_rate}.
Brackets give $95\%$ bootstrap confidence intervals from $50{,}000$ replicates.
}
\label{tab:surface_code_p0p001}
\begin{ruledtabular}
\begin{tabular}{c r r c}
$d$ & Summed shots & Summed failures &
$\overline{P}_{\mathrm{L}}$ [95\% CI] \\
\hline
3 & $80{,}000{,}000$ & $65{,}418$ &
$8.17725\times10^{-4}$ $[8.11412,8.24013]\times10^{-4}$ \\
5 & $800{,}000{,}000$ & $111{,}543$ &
$1.39429\times10^{-4}$ $[1.38624,1.40239]\times10^{-4}$ \\
7 & $8{,}000{,}000{,}000$ & $151{,}354$ &
$1.89193\times10^{-5}$ $[1.88245,1.90149]\times10^{-5}$ \\
\end{tabular}
\end{ruledtabular}
\end{table*}

\subsection{Hotspot identification}
\label{app:hotspot_aggregation}

To calculate the sensitivity of each physical component, we group noise components by the physical qubits and directed CNOT couplings on which they act.
All one-qubit noise channels acting on the same physical qubit are assigned to that qubit, and all two-qubit noise channels acting on the same directed CNOT coupling are assigned to that coupling.
For each physical component, we sum the sensitivities of all associated noise components across all syndrome rounds according to Eq.~\eqref{eq:physical_component_sensitivity}.
The $X$- and $Z$-memory experiments therefore have the same set of physical components and the same component labels.
Table~\ref{tab:surface_code_inventory} summarizes the numbers of
qubits, detectors, and physical components in these circuits.

\begin{table*}[t]
\caption{
Numbers of qubits, detectors, and physical components in rotated surface code memories with $d$ syndrome rounds.
The counts are the same for the $X$ and $Z$ memories.
Physical components include the data qubits, check ancillas, and directed CNOT couplings, each counted once across all rounds.
}
\label{tab:surface_code_inventory}
\begin{ruledtabular}
\begin{tabular}{c r r r r r}
$d$ & Data & Check ancillas & Physical qubits & Detectors & Physical components \\
\hline
3 & 9  & 8  & 17 & 24  & 41  \\
5 & 25 & 24 & 49 & 120 & 129 \\
7 & 49 & 48 & 97 & 336 & 265 \\
\end{tabular}
\end{ruledtabular}
\end{table*}

We estimate the signed sensitivity of each physical component $c$ using all Monte Carlo samples within each memory basis.
We denote these estimates by $S_c^{X}$ for the $X$ memory and $S_c^{Z}$ for the $Z$ memory.
We identify hotspots using the magnitude of the sensitivity averaged over the $X$ and $Z$ memories:
\begin{equation}
    A_c
    =
    \left|
    \frac{S_c^{X}+S_c^{Z}}{2}
    \right|.
    \label{eq:surface_code_basis_averaged_sensitivity}
\end{equation}
Fig.~\ref{fig:surface_code_spatial_sensitivity} displays the within-distance normalized value
\begin{equation}
    \widetilde A_c^{(d)}
    =
    \frac{A_c}{\max_{c'\in\mathcal C_d}A_{c'}},
    \label{eq:surface_code_relative_sensitivity}
\end{equation}
where $\mathcal C_d$ is the component set at distance $d$.
This normalization puts all panels on a common $0$--$1$ scale while preserving the component ranking at each code distance.

We calculate the risk shares reported in Sec.~\ref{sec:surface_code_results} using the relation $\RL_c=p_c\SL_c$.
Since all physical components have the same error probability $p$, the linearized marginal risk of component $c$ is
\begin{equation}
    \overline{\RL}_{c}
    =
    p\frac{S_c^{X}+S_c^{Z}}{2}.
    \label{eq:component_basis_averaged_signed_risk}
\end{equation}
For a set $G\subseteq\mathcal C_d$, we sum the marginal risks of its components and divide by the total over all physical components:
\begin{equation}
    q_G
    =
    \frac{
        \sum_{c\in G}\overline{\RL}_{c}
    }{
        \sum_{c\in\mathcal C_d}\overline{\RL}_{c}
    }.
    \label{eq:signed_risk_share}
\end{equation}
Taking $G$ to contain all data qubits, all check ancillas, or all directed CNOT couplings gives the corresponding shares.

At $d=3,5,7$, we select the $3$, $7$, and $13$ components with the largest sensitivity magnitudes $A_c$, respectively.
These components make up $7.32\%$, $5.43\%$, and $4.91\%$ of all physical components in the corresponding circuits.
Using each selected set as $G$ in Eq.~\eqref{eq:signed_risk_share}, we find that these components account for $15.12\%$, $10.24\%$, and $8.93\%$ of the total linearized marginal risk, respectively.

\subsection{Hotspot structure}
\label{app:hotspot_geometry}

This subsection analyzes the central concentration and diagonal tendency of the surface code hotspots in Fig.~\ref{fig:surface_code_spatial_sensitivity}.
We first relate local errors to logical failure, then analyze boundary effects in an ideal model to explain central concentration, and finally examine how CNOT ordering contributes to the diagonal tendency.

\paragraph{Local errors and logical failure.}
A local error can cause a logical failure by combining with other errors and the decoder's recovery to change the logical outcome.
Let $e$ denote a sampled shot with the local error $f$ absent.
The corresponding syndrome and final readout records are $\mathbf{m}(e)$ and $\mathbf{y}(e)$, respectively, and the logical failure indicator is
\begin{equation}
\begin{aligned}
    \mathcal L(e)
    &=L(\mathbf{m}(e),\mathbf{y}(e),\decoder)\\
    &=g(\mathbf{y}(e))\oplus\delta_{\decoder}(\mathbf{m}(e))
      \oplus b_{\mathrm{ideal}}.
\end{aligned}
\end{equation}
Let $\mathbf d_f$ specify the syndrome bits flipped by $f$, and let $o_f$ indicate whether it flips the readout bit $g(\mathbf{y}(e))$.
Define
\begin{equation}
    K_f(\mathbf{m})=o_f\oplus\delta_{\decoder}(\mathbf{m})
    \oplus\delta_{\decoder}(\mathbf{m}\oplus\mathbf d_f).
    \label{eq:hotspot_fault_toggle}
\end{equation}
Write $e+f$ for the error configuration after adding $f$, with all other errors unchanged.
Substituting the syndrome and readout flips gives
\begin{equation}
\begin{aligned}
    \mathcal L(e+f)
    &=g(\mathbf{y}(e))\oplus o_f\\
    &\quad\oplus\delta_{\decoder}(\mathbf{m}(e)\oplus\mathbf d_f)
      \oplus b_{\mathrm{ideal}}\\
    &=\mathcal L(e)\oplus K_f(\mathbf{m}(e)).
\end{aligned}
\end{equation}
Thus $K_f(\mathbf{m}(e))=1$ means that adding $f$ changes the decoded logical outcome, introducing or removing a logical failure.
For an independent Bernoulli error and a fixed decoder, Eq.~\eqref{eq:conditional_failure_sensitivity} gives
\begin{equation}
\begin{aligned}
    S_f
    &=\E_e[\mathcal L(e+f)-\mathcal L(e)]\\
    &=\E_e[(1-2\mathcal L(e))K_f(\mathbf{m}(e))].
\end{aligned}
    \label{eq:hotspot_pivotal_sensitivity}
\end{equation}
The expectation is over shots with $f$ absent.

\paragraph{Boundary effects and central concentration.}
To analyze the central concentration, we count the combinations of errors that cause logical failure and examine how the boundaries restrict these combinations.
We consider an idealized model of a rotated surface code of distance $d=2t+1$, $t\ge1$, with perfect syndrome measurements and minimum-Hamming-weight decoding.
Only data qubits are noisy, and we analyze $X$ errors and $Z$ errors separately.
For a fixed choice $P\in\{X,Z\}$, each data qubit independently has a $P$ error with probability $p$ and remains unchanged with probability $1-p$.

Let $E$ be the set of data qubits carrying a $P$ error.
We denote by $\partial E$ the set of detectors flipped by an odd number of errors in $E$; an even number of errors that flip the same detector leave its outcome unchanged.
Let $C$ be the set of data qubits on which the decoder applies $P$ to correct the errors.
After the error and recovery, a qubit in both $E$ and $C$ has undergone $P$ twice, giving $P^2=I$, whereas a qubit in only one set has a residual $P$ error.
We denote the support of this residual Pauli operator by $E+C=(E\setminus C)\cup(C\setminus E)$.

We call an error support \emph{malignant} if it causes a logical failure under the specified decoder.
Minimum-weight decoding corrects every configuration with $|E|\le t$, so logical failure can first occur at $|E|=t+1$.
To compare locations, let $a_P(q)$ count the malignant configurations with exactly $t+1$ errors of type $P$, including an error at data qubit $q$.
Removing the error at $q$ from such a configuration leaves $t$ errors that the decoder can correct; adding it back causes logical failure.
At uniform error probability $p$, these configurations with $t$ errors give the leading contribution to the ideal model sensitivity:
\begin{equation}
    S_{q,\mathrm{ideal}}^{(P)}=a_P(q)p^t+O(p^{t+1}).
    \label{eq:hotspot_ideal_sensitivity_expansion}
\end{equation}
Here $P$ labels the physical error type, not the memory basis.

To identify the configurations counted by $a_P(q)$, we show that an error support $E$ with $|E|=t+1$ is malignant exactly when it is contained in the support $\Gamma$ of a weight-$d$ nontrivial logical operator of type $P$.
Here $\Gamma$ is a set of exactly $d$ data qubits.
Applying $P$ to these qubits leaves all stabilizer outcomes unchanged but acts on the encoded qubit as $\overline X$ for $P=X$ or $\overline Z$ for $P=Z$.

First suppose that $E$ causes a logical failure, and let $C$ be the recovery selected by the decoder.
The error support $E$ itself is a candidate recovery with the observed syndrome, so minimum-weight decoding guarantees $|C|\leq |E|=t+1$.
Failure means that the residual operator, supported on $E+C$, is a nontrivial logical Pauli of type $P$.
Because this code encodes one logical qubit, every nontrivial logical operator of type $P$ can be obtained from a weight-$d$ logical string by multiplying stabilizer generators of the same Pauli type.
Each such generator acts on $w=2$ or $4$ qubits.
If it overlaps the current operator on $k$ qubits, multiplication cancels the two $P$ operations on those $k$ qubits and adds a $P$ operation on each of the other $w-k$ qubits.
The weight therefore changes by $(w-k)-k=w-2k$, which is even.
The residual operator and the weight-$d$ logical string therefore have the same weight parity.
Since $d=2t+1$ is odd, $|E+C|$ is odd.
Using
\begin{equation}
    |E+C|=|E|+|C|-2|E\cap C|,
\end{equation}
we see that $|E|$ and $|C|$ must have opposite parity.
In particular, $|C|$ cannot equal $|E|=t+1$, so $|C|\le t$.
On the other hand, any nontrivial logical operator has weight at least $d$.
Combining these bounds gives
\begin{equation}
    d\le |E+C|\le |E|+|C|\le 2t+1=d,
    \label{eq:hotspot_minimum_failure_distance}
\end{equation}
Every inequality must therefore be an equality.
In particular, $|C|=t$ and $|E+C|=|E|+|C|$, which requires $E\cap C=\varnothing$.
Thus $\Gamma=E\cup C$ is a weight-$d$ nontrivial logical support containing $E$.

Conversely, suppose that $E\subseteq\Gamma$, where $\Gamma$ is a weight-$d$ nontrivial logical support.
The complementary support $C_0=\Gamma\setminus E$ has weight $t$.
Since a logical operator produces no detector events, $\partial\Gamma=0$, and hence
\begin{equation}
    \partial C_0=\partial E,\qquad E+C_0=\Gamma.
    \label{eq:hotspot_complement_syndrome}
\end{equation}
Thus $C_0$ is an admissible recovery of weight $t$: it has the observed syndrome, although applying it after $E$ leaves the logical error $\Gamma$.
Because the decoder minimizes weight, its chosen recovery $C$ must satisfy $|C|\le |C_0|=t$.

Since $C$ and $C_0$ have the same syndrome, their symmetric difference $C+C_0$ has zero syndrome.
Its weight satisfies
\begin{equation}
    |C+C_0|\le |C|+|C_0|\le 2t<d.
\end{equation}
A zero-syndrome Pauli operator is either a stabilizer, which acts as the identity on the code space, or a nontrivial logical operator.
The latter must have weight at least $d$, so the operator supported on $C+C_0$ must be a stabilizer.

The actual residual support can be written as
\begin{equation}
    E+C=(E+C_0)+(C_0+C)=\Gamma+(C_0+C),
\end{equation}
Thus the actual residual operator is the logical operator supported on $\Gamma$ multiplied by a stabilizer, and therefore has the same nontrivial logical action.
Every minimum-weight choice of recovery therefore causes a logical failure for this $E$.

This criterion reduces the counting problem to finding sets of error locations contained in the support of a single weight-$d$ nontrivial logical operator of type $P$.
Label the data qubit positions by indices $(x,y)\in\{0,\ldots,2t\}^2$, with $x$ increasing to the right and $y$ downward, as in Fig.~\ref{fig:surface_code_check_circuits}.
The central data qubit is at $(t,t)$.
For $P=X$, let $\Gamma_X$ be the support of a weight-$d$ logical $X$ operator.
Applying $Z$ to every data qubit in any horizontal row implements a logical $\overline Z$, which must anticommute with the logical $X$ operator supported on $\Gamma_X$.
Anticommutation therefore requires $\Gamma_X$ to intersect each row in an odd number of qubits.
There are $d$ disjoint rows and $|\Gamma_X|=d$, so $\Gamma_X$ contains exactly one qubit in each row.
For $P=Z$, every vertical column supports a logical $\overline X$.
The same anticommutation argument shows that the support $\Gamma_Z$ of a weight-$d$ logical $Z$ operator contains exactly one qubit in each column.

In the layout used here, the bulk check between columns $x,x+1$ and rows $y,y+1$ is of type $Z$ when $x+y$ is even and of type $X$ when $x+y$ is odd.
Fix a physical $X$ error at $(x,y)$ and consider a weight-$d$ logical $X$ operator whose support $\Gamma_X$ contains this qubit.
The qubit of $\Gamma_X$ in row $y+1$ must be at one of the two positions marked by question marks in Fig.~\ref{fig:hotspot_x_reachability}(a,b).
In the bulk, the $Z$ check immediately below $(x,y)$ already overlaps $\Gamma_X$ once in row $y$.
For the complete logical $X$ operator to commute with this check, $\Gamma_X$ must also contain one of its two corners in row $y+1$.
For even $x+y$, these corners are $(x,y+1)$ and $(x+1,y+1)$; for odd $x+y$, they are $(x-1,y+1)$ and $(x,y+1)$.
The rule above describes two adjacent rows, corresponding to $g=1$.
Applying this rule $g$ times, with $g\ge1$, gives the allowed positions $(x',y+g)$ for a second $X$ error.
The two $X$ errors at $(x,y)$ and $(x',y+g)$ can lie in the support of the same weight-$d$ logical $X$ operator if and only if
\begin{equation}
    x'\in
    \begin{cases}
        [x-g+1,x+g],&x+y\text{ even},\\
        [x-g,x+g-1],&x+y\text{ odd}.
    \end{cases}
    \label{eq:hotspot_x_error_reachability}
\end{equation}
Figure~\ref{fig:hotspot_x_reachability}(c,d) illustrates this repeated application of the single-row rule.
Each allowed connection can be extended vertically to the top and bottom boundaries to form the support of a weight-$d$ logical $X$ operator.

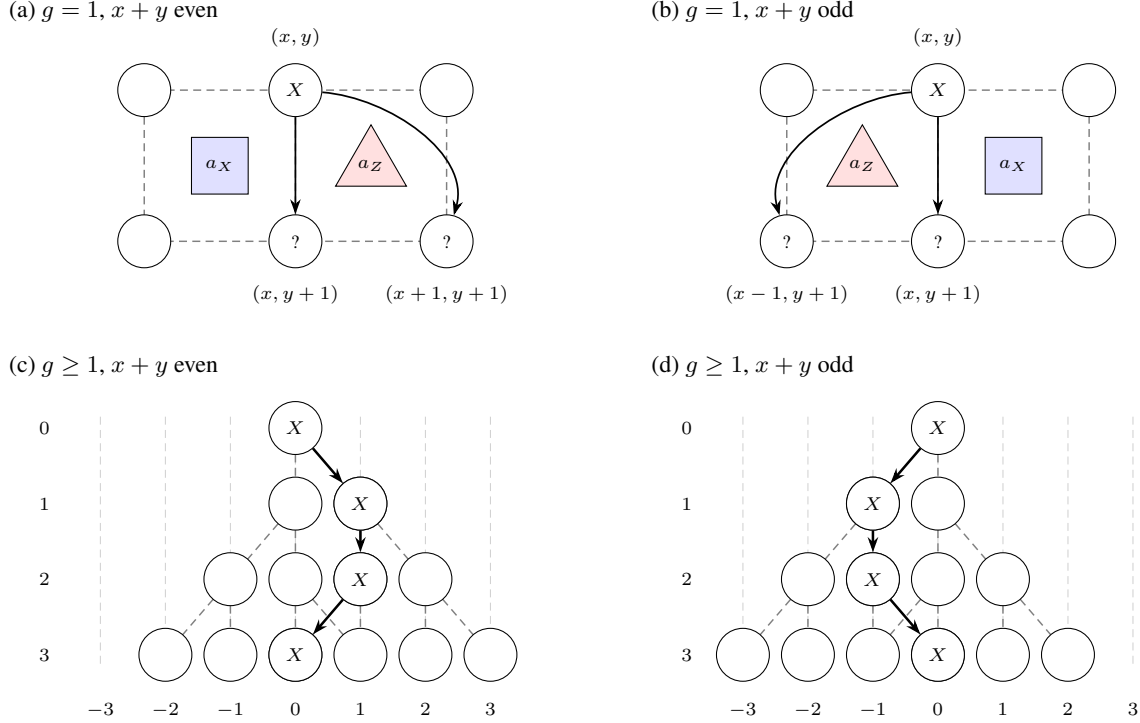
\begin{figure*}[t]
    \centering
    \begin{tikzpicture}[
        x=1cm,y=-1cm,
        font=\scriptsize,
        data/.style={circle,draw=black,fill=white,minimum size=7mm,
                     inner sep=0pt,font=\scriptsize},
        xcheck/.style={rectangle,draw=black,fill=blue!12,minimum size=7.5mm,
                       inner sep=0pt,font=\scriptsize},
        zcheck/.style={regular polygon,regular polygon sides=3,draw=black,
                       fill=red!12,minimum size=10mm,inner sep=0.5pt,
                       font=\scriptsize},
        continuation/.style={-{Stealth[length=2mm,width=1.4mm]},line width=0.65pt},
        example/.style={-{Stealth[length=2mm,width=1.4mm]},line width=0.95pt},
        panel/.style={font=\small,anchor=west},
        note/.style={font=\scriptsize,align=center},
        support/.style={gray,densely dashed,line width=0.55pt},
        candidate/.style={gray,densely dashed,line width=0.55pt,
                          shorten <=3.5mm,shorten >=3.5mm},
        guide/.style={gray!35,densely dashed,line width=0.4pt}
    ]
        \path[use as bounding box] (0,-.05) rectangle (16.6,9.65);

        \begin{scope}[shift={(4.0,1.25)}]
            \node[panel] at (-3.9,-1.05) {(a) $g=1$, $x+y$ even};
            \draw[support] (-2,0) rectangle (2,2);
            \draw[support] (0,0)--(0,2);
            \node[xcheck] at (-1,1) {$a_X$};
            \node[zcheck] at (1,1) {$a_Z$};
            \node[data] (reachAS) at (0,0) {$X$};
            \node[data] at (-2,0) {};
            \node[data] at (2,0) {};
            \node[data] at (-2,2) {};
            \node[data] (reachAL) at (0,2) {?};
            \node[data] (reachAR) at (2,2) {?};
            \draw[continuation] (reachAS)--(reachAL);
            \draw[continuation] (reachAS)
                .. controls (1.25,.1) and (2.32,.85) .. (reachAR);
            \node[note,anchor=south] at (0,-.46) {$(x,y)$};
            \node[note,anchor=north] at (0,2.46) {$(x,y+1)$};
            \node[note,anchor=north] at (2,2.46) {$(x+1,y+1)$};
        \end{scope}

        \begin{scope}[shift={(12.5,1.25)}]
            \node[panel] at (-3.9,-1.05) {(b) $g=1$, $x+y$ odd};
            \draw[support] (-2,0) rectangle (2,2);
            \draw[support] (0,0)--(0,2);
            \node[zcheck] at (-1,1) {$a_Z$};
            \node[xcheck] at (1,1) {$a_X$};
            \node[data] (reachBS) at (0,0) {$X$};
            \node[data] at (-2,0) {};
            \node[data] at (2,0) {};
            \node[data] (reachBL) at (-2,2) {?};
            \node[data] (reachBR) at (0,2) {?};
            \node[data] at (2,2) {};
            \draw[continuation] (reachBS)--(reachBR);
            \draw[continuation] (reachBS)
                .. controls (-1.25,.1) and (-2.32,.85) .. (reachBL);
            \node[note,anchor=south] at (0,-.46) {$(x,y)$};
            \node[note,anchor=north] at (-2,2.46) {$(x-1,y+1)$};
            \node[note,anchor=north] at (0,2.46) {$(x,y+1)$};
        \end{scope}

        \node[panel] at (.1,4.9) {(c) $g\geq1$, $x+y$ even};
        \node[panel] at (8.6,4.9) {(d) $g\geq1$, $x+y$ odd};

        \foreach \ox/\parity in {4.0/0,12.5/1} {
            \begin{scope}[shift={(\ox,5.72)},x=.86cm,y=-1.0cm]
                \foreach \k in {-3,-2,-1,0,1,2,3} {
                    \draw[guide] (\k,-.15)--(\k,3.12);
                    \node[note,anchor=north] at (\k,3.52) {$\k$};
                }
                \foreach \g in {0,1,2,3} {
                    \node[note,anchor=east] at (-3.65,\g) {$\g$};
                }

                \foreach \g in {0,1,2} {
                    \ifnum\g=0
                        \def\kmin{0}\def\kmax{0}
                    \else
                        \pgfmathtruncatemacro{\kmin}{-\g+1-\parity}
                        \pgfmathtruncatemacro{\kmax}{\g-\parity}
                    \fi
                    \foreach \k in {\kmin,...,\kmax} {
                        \pgfmathtruncatemacro{\r}{mod(\parity+\k+\g+20,2)}
                        \pgfmathtruncatemacro{\nextg}{\g+1}
                        \pgfmathtruncatemacro{\nextk}{\k+1-2*\r}
                        \draw[candidate] (\k,\g)--(\k,\nextg);
                        \draw[candidate] (\k,\g)--(\nextk,\nextg);
                    }
                }

                \node[data] (reachRoot\parity) at (0,0) {$X$};
                \foreach \g in {1,2,3} {
                    \pgfmathtruncatemacro{\kmin}{-\g+1-\parity}
                    \pgfmathtruncatemacro{\kmax}{\g-\parity}
                    \foreach \k in {\kmin,...,\kmax}
                        \node[data] at (\k,\g) {};
                }
                \ifnum\parity=0
                    \node[data] (reachFirst\parity) at (1,1) {$X$};
                    \node[data] (reachSecond\parity) at (1,2) {$X$};
                \else
                    \node[data] (reachFirst\parity) at (-1,1) {$X$};
                    \node[data] (reachSecond\parity) at (-1,2) {$X$};
                \fi
                \node[data] (reachLast\parity) at (0,3) {$X$};
                \draw[example] (reachRoot\parity)--(reachFirst\parity);
                \draw[example] (reachFirst\parity)--(reachSecond\parity);
                \draw[example] (reachSecond\parity)--(reachLast\parity);
            \end{scope}
        }
    \end{tikzpicture}
    \caption{
    Allowed positions for $X$ errors in the support of a weight-$d$ logical $X$ operator, with $x$ increasing to the right and $y$ downward.
    As in Fig.~\ref{fig:surface_code_check_circuits}, circles denote data qubits, blue squares $X$-check ancillas, and red triangles $Z$-check ancillas.
    (a,b) For a fixed $X$ error at $(x,y)$, question marks show the two choices in row $y+1$ that allow the logical operator to commute with the $Z$ check below.
    (c,d) Applying the same condition over $g=1,2,3$ rows gives the allowed column offsets for even and odd $x+y$, respectively.
    }
    \label{fig:hotspot_x_reachability}
\end{figure*}

For $P=Z$, applying the same even-overlap requirement to the $X$ checks between adjacent columns gives the condition for two $Z$ errors at $(x,y)$ and $(x+g,y')$:
\begin{equation}
    y'\in
    \begin{cases}
        [y-g,y+g-1],&x+y\text{ even},\\
        [y-g+1,y+g],&x+y\text{ odd}.
    \end{cases}
    \label{eq:hotspot_error_reachability}
\end{equation}
All intervals are restricted to integer positions in $[0,2t]$.
Each interval contains $2g$ positions before boundary truncation.

When the fixed $X$ error is near the left or right boundary, the intervals above can include positions outside these boundaries, which are unavailable; the top and bottom boundaries have the same effect for $Z$ errors.
This leaves fewer combinations of the remaining errors that can cause logical failure together with the fixed error.
When the fixed error is at the center, more combinations of the remaining errors lie within the code boundaries and cause logical failure together with it; the following proposition establishes this comparison along the middle row for $X$ errors and the middle column for $Z$ errors.

\begin{proposition}[Transverse central enhancement]
In this ideal model, $a_X(q)$ is strictly largest at the center along the middle row, and $a_Z(q)$ is strictly largest at the center along the middle column:
\begin{equation}
\begin{aligned}
    a_X(x,t)&<a_X(t,t),\qquad x\ne t,\\
    a_Z(t,y)&<a_Z(t,t),\qquad y\ne t.
\end{aligned}
    \label{eq:hotspot_transverse_central_bound}
\end{equation}
\end{proposition}

\begin{proof}
For $P=X$, $a_X(x,t)$ counts malignant configurations with $t+1$ $X$ errors, one of which is fixed at $(x,t)$.
The remaining $t$ errors must occupy $t$ distinct rows other than $y=t$, because the support of a weight-$d$ logical $X$ operator contains exactly one qubit per row.
For each choice of these rows, we count the combinations of column positions that allow all $t+1$ errors to lie in the support of a single weight-$d$ logical $X$ operator.
Summing over the row choices gives $a_X(x,t)$.
Without the left and right boundaries, each gap of $g$ rows between consecutive errors gives $2g$ column choices for the next error as we count upward or downward from the fixed error, so the total count for any fixed choice of rows is independent of $x$.
At $x=t$, all these configurations lie within the code boundaries: each error is no farther from the center horizontally than vertically, and the vertical distance is at most $t$.
Thus the center attains the unrestricted count for every choice of occupied rows, while other positions can only lose configurations.
For $x<t$, place the remaining $t$ errors along the allowed diagonal extending left from $(x,t)$, one in each successive row.
The last error is at column $x-t<0$, outside the left boundary, so this configuration is excluded.
For $x>t$, diagonal steps towards increasing $x$ give the corresponding exclusion.
At least one choice of occupied rows therefore loses a configuration, proving the strict inequality for $a_X$.

For $P=Z$, $a_Z(t,y)$ counts malignant configurations with $t+1$ $Z$ errors, including one at $(t,y)$.
For each choice of $t$ distinct columns other than $x=t$, we count the combinations of row positions that allow all $t+1$ errors to lie in the support of a single weight-$d$ logical $Z$ operator.
With the top and bottom boundaries removed, each successive column gap $g$ admits $2g$ row choices, so the unrestricted count is independent of $y$.
At $y=t$, every such error configuration lies within the code boundaries: its vertical displacement from the fixed error on either side is bounded by the corresponding horizontal span, which is at most $t$.
If $y<t$, choose the horizontal direction that permits $t$ consecutive diagonal steps towards decreasing $y$; the last error lies at row $y-t<0$, outside the top boundary, so this configuration is excluded.
If $y>t$, use steps towards increasing $y$.
This proves the strict inequality for $a_Z$.
\end{proof}

\paragraph{CNOT ordering and diagonal tendency.}
To analyze the diagonal tendency, we examine how errors at different CNOT layers propagate to the data qubits.
Consider a bulk check with data neighbors $q_1,\ldots,q_4$ ordered by CNOT layer, as shown in Fig.~\ref{fig:surface_code_check_circuits}(a,c) and specified in Eqs.~\eqref{eq:surface_code_x_check_order} and \eqref{eq:surface_code_z_check_order}.
Write $P_j$ for Pauli $P$ on $q_j$.
A CNOT propagates an $X$ error from the control to the target ($X_c\mapsto X_cX_t$) and a $Z$ error from the target to the control ($Z_t\mapsto Z_cZ_t$).
Thus an ancilla $X$ error in an $X$ check, or an ancilla $Z$ error in a $Z$ check, spreads to the data qubits involved in the remaining CNOTs.
Immediately after the layer-$j$ CNOT, $P_aI_{q_j}$ represents a $P$ error on the ancilla alone, while $P_aP_j$ represents $P$ errors on both the ancilla and $q_j$.
Propagation through the remaining CNOTs of the check then produces the data errors $H_j$ and $P_jH_j$, respectively, where
\begin{equation}
    H_j=\prod_{k=j+1}^{4}P_k.
    \label{eq:hotspot_hook_suffix}
\end{equation}
Modulo the measured stabilizer $K_a=P_1P_2P_3P_4$, the two representatives $E$ and $EK_a$ have complementary data supports.
Table~\ref{tab:hotspot_hook_cosets} lists the resulting classes $[E]=\{E,EK_a\}$ and their minimum weights.

\begin{table}[b]
\caption{
Data errors produced by Pauli errors propagating through the CNOTs following layer $j$ of a bulk stabilizer measurement.
The columns $[H_j]$ and $[P_jH_j]$ correspond to the errors $P_aI_{q_j}$ and $P_aP_j$, respectively, after propagation through the remaining CNOTs.
Brackets group errors that differ by the measured stabilizer $K_a$.
The number in parentheses gives the smallest number of affected data qubits in each group.
}
\label{tab:hotspot_hook_cosets}
\begin{ruledtabular}
\begin{tabular}{c c c}
$j$ & $[H_j]$ & $[P_jH_j]$ \\
\hline
$1$ & $[P_1]\ (1)$       & $[I]\ (0)$ \\
$2$ & $[P_3P_4]\ (2)$    & $[P_1]\ (1)$ \\
$3$ & $[P_4]\ (1)$       & $[P_3P_4]\ (2)$ \\
$4$ & $[I]\ (0)$         & $[P_4]\ (1)$
\end{tabular}
\end{ruledtabular}
\end{table}

Among the errors $P_aI_{q_j}$ and $P_aP_j$ in layers $j=1,\ldots,4$, only the layer-$2$ error $P_aI_{q_2}$ and the layer-$3$ error $P_aP_3$ yield weight-two supports irreducible modulo $K_a$.
For both check types, the CNOTs in layers $2$ and $3$ connect the ancilla to the pair of data qubits at offsets $(-1,+1)$ and $(+1,-1)$, while those in layers $1$ and $4$ connect it to the pair at $(+1,+1)$ and $(-1,-1)$.
In the matrix-style coordinates of Fig.~\ref{fig:surface_code_check_circuits}, these coupling pairs lie on the displayed $/$ and $\backslash$ diagonals, respectively.
For the two Pauli error types considered here, errors after the $/$-oriented CNOTs can propagate into irreducible two-qubit data errors, whereas those after the $\backslash$-oriented CNOTs produce data errors of minimum weight zero or one.
Intuitively, CNOT propagation may make local errors on data qubits along the $/$ direction more likely to combine with other errors and cause logical failure, giving rise to the observed hotspots.

\section{Hotspot-guided targeted intervention simulations}
\label{app:hotspot_intervention}

\subsection{Intervention setup and sampling}

At $p=0.001$, we compare circuits without intervention with circuits subject to either hotspot-guided or random intervention.
For the interventions in Fig.~\ref{fig:hotspot_targeted_intervention}, we use the hotspot sets identified in Appendix~\ref{app:hotspot_aggregation}.
The local error rate of each selected physical component is reduced according to Eq.~\eqref{eq:targeted_intervention}; all other component error rates remain unchanged.

At code distance $d$, let $k_d$ denote the number of components in the hotspot set and $N_d$ the total number of physical components in the circuit.
For comparison, we generate 40 distinct random sets, each containing $k_d$ components chosen uniformly at random from all $N_d$ components.
For each random set, we halve the local error rate of every selected component ($\alpha=0.5$), as in the hotspot-guided intervention.

The decoder is constructed from the detector error model of the circuit being evaluated, so its edge weights reflect the noise probabilities after intervention.

For each code distance, we simulate the circuit before and after hotspot-guided intervention in both the $X$ and $Z$ memory bases.
At $d=3$, we use $6.48\times10^7$ shots for the $X$ memory and $7.02\times10^7$ for the $Z$ memory both before and after intervention, totaling $1.35\times10^8$ shots in each case.
At $d=5$, the corresponding numbers are $3.726\times10^8$ and $4.104\times10^8$, totaling $7.83\times10^8$ shots in each case.
At $d=7$, they are $2.6784\times10^9$ and $3.0618\times10^9$, totaling $5.7402\times10^9$ shots in each case.

For the random intervention, each of the 40 component sets uses $1.2\times10^6$ shots for the $X$ memory and $1.3\times10^6$ for the $Z$ memory at $d=3$.
The corresponding numbers per set are $6.9\times10^6$ and $7.6\times10^6$ at $d=5$, and $4.96\times10^7$ and $5.67\times10^7$ at $d=7$.
The random intervention therefore uses $1.0\times10^8$, $5.8\times10^8$, and $4.252\times10^9$ shots in total at $d=3,5,7$, respectively.

\subsection{Statistical analysis}
\label{app:unified_intervention_statistics}

We use the same statistical analysis for the surface code and mitten code intervention simulations.
We use $s=0$, $s=\mathrm{target}$, and $s=\mathrm{random}$ to denote no intervention, hotspot-guided intervention, and intervention on a fixed random component set, respectively.
For memory basis $b\in\{X,Z\}$, let $F_b^{(s)}$ and $N_b^{(s)}$ denote the total numbers of logical failures and shots for circuit $s$, respectively.
We estimate the logical error rate in each basis as $F_b^{(s)}/N_b^{(s)}$ and average the two estimates:
\begin{equation}
    P_{\mathrm{L}}^{(s)}
    =
    \frac{1}{2}\left(
      \frac{F_X^{(s)}}{N_X^{(s)}}+
      \frac{F_Z^{(s)}}{N_Z^{(s)}}
    \right).
    \label{eq:intervention_basis_average_ler}
\end{equation}

For random component set $j=1,\ldots,40$, we write its logical error rate averaged over the two memory bases as $P_{\mathrm{L},j}^{(\mathrm{random})}$.
The reported random-intervention rate is the average over the 40 sets:
\begin{equation}
    \overline{P}_{\mathrm{L}}^{(\mathrm{random})}
    =\frac{1}{40}\sum_{j=1}^{40}
      P_{\mathrm{L},j}^{(\mathrm{random})}.
\end{equation}
This estimates the mean performance of uniform random component selection.
Relative to the circuit without intervention, the reduction in logical error rate from hotspot-guided intervention and the mean reduction from random intervention are given by
\begin{equation}
    \begin{aligned}
        \Delta_{\mathrm{t}}
        &=1-\frac{P_{\mathrm{L}}^{(\mathrm{target})}}{P_{\mathrm{L}}^{(0)}},\\
        \Delta_{\mathrm{r}}
        &=1-\frac{\overline{P}_{\mathrm{L}}^{(\mathrm{random})}}{P_{\mathrm{L}}^{(0)}}.
    \end{aligned}
\end{equation}

We use bootstrap resampling to quantify the statistical uncertainty in the logical error rates and their reductions.
In each bootstrap replicate, we independently generate a failure count $F_b^{(s)*}$ for each memory basis in the simulations before and after hotspot-guided intervention.
We treat the $N_b^{(s)}$ shots as independent samples, each with failure probability $F_b^{(s)}/N_b^{(s)}$:
\begin{equation}
    F_b^{(s)*}\sim\operatorname{Binomial}\left(
        N_b^{(s)},\frac{F_b^{(s)}}{N_b^{(s)}}
    \right),
    \quad s\in\{0,\mathrm{target}\}.
    \label{eq:intervention_bootstrap_counts}
\end{equation}
We then average the two rates $F_X^{(s)*}/N_X^{(s)}$ and $F_Z^{(s)*}/N_Z^{(s)}$ for each case.

For random intervention, the 40 simulated component sets give 40 logical error rate estimates $P_{\mathrm{L},j}^{(\mathrm{random})}$, each averaged over the $X$ and $Z$ memories.
In each bootstrap replicate, we sample 40 of these estimates independently with replacement and average them to obtain the random-intervention logical error rate.

In each bootstrap replicate, we use the same resampled rate for the circuit without intervention to calculate the reductions from hotspot-guided and random intervention.
We obtain differences in logical error rate by subtracting the corresponding logical error rates from that replicate.
The 95\% confidence intervals for the logical error rates, reductions, and differences in logical error rate are given by the 2.5th and 97.5th percentiles of the corresponding estimates over 50,000 bootstrap replicates.

The improvement ratio reported in Sec.~\ref{sec:hotspot_intervention} is
\begin{equation}
    \mathcal I
    =\frac{\Delta_{\mathrm{t}}}{\Delta_{\mathrm{r}}}
    =\frac{1-P_{\mathrm{L}}^{(\mathrm{target})}/P_{\mathrm{L}}^{(0)}}
           {1-\overline{P}_{\mathrm{L}}^{(\mathrm{random})}/P_{\mathrm{L}}^{(0)}}.
    \label{eq:targeted_vs_random_improvement_ratio}
\end{equation}
We use the same Fieller method~\cite{fieller1954some} to calculate the confidence intervals for the improvement ratios in the surface code and mitten code simulations.
Fieller intervals are used because the denominator is the reduction under random intervention, which may not be distinguishable from zero.
Let $A=P_{\mathrm L}^{(0)}-P_{\mathrm L}^{(\mathrm{target})}$ and $B=P_{\mathrm L}^{(0)}-\overline P_{\mathrm L}^{(\mathrm{random})}$, so that $\mathcal I=A/B$.
We denote the estimates from the simulation data by $\widehat A$ and $\widehat B$, and their values in each bootstrap replicate by $A^*$ and $B^*$.
In each replicate, we use the same bootstrap estimate of the baseline rate, $P_{\mathrm L}^{(0)*}$, to calculate both $A^*$ and $B^*$.

We approximate the joint sampling distribution of $\widehat A$ and $\widehat B$ by a normal distribution.
For a candidate improvement ratio $r$, we check whether $A-rB$ is consistent with zero.
The $95\%$ confidence set contains all values $r$ satisfying
\begin{equation}
    (\widehat A-r\widehat B)^2
    \le z_{0.975}^{\,2}\,
    \widehat{\mathrm{Var}}_{*}(A^*-rB^*).
    \label{eq:intervention_fieller}
\end{equation}
Here $z_{0.975}\approx1.96$ is the 97.5th percentile of the standard normal distribution.
The variance on the right-hand side is the sample variance of $A^*-rB^*$ across the 50,000 bootstrap replicates.

Table~\ref{tab:hotspot_intervention_results} summarizes the surface code intervention results.

\begin{table*}[t]
\caption{
Logical error rates and relative reductions for rotated surface code interventions at $p=0.001$ and $\alpha=0.5$, corresponding to Fig.~\ref{fig:hotspot_targeted_intervention}.
$k_d$ and $N_d$ denote the numbers of selected and total physical components.
All rates average the $X$- and $Z$-memory results.
$P_{\mathrm{L}}^{(0)}$ and $P_{\mathrm{L}}^{(\mathrm{t})}$ are the rates without intervention and after hotspot-guided intervention.
$\overline{P}_{\mathrm{L}}^{(\mathrm{r})}$ is the mean rate over 40 sets of $k_d$ components chosen uniformly at random.
$\Delta_{\mathrm{t}}$ and $\Delta_{\mathrm{r}}$ give the corresponding relative reductions to $P_{\mathrm{L}}^{(0)}$, and $\mathcal I=\Delta_{\mathrm{t}}/\Delta_{\mathrm{r}}$ is their ratio.
Brackets give $95\%$ confidence intervals: bootstrap percentile intervals for the reductions and Fieller intervals for $\mathcal I$.
}
\label{tab:hotspot_intervention_results}
\begin{ruledtabular}
\setlength{\tabcolsep}{2.8pt}
\begin{tabular}{c c c c c c c c}
$d$ & $k_d/N_d$ & $P_{\mathrm{L}}^{(0)}$ &
$\overline{P}_{\mathrm{L}}^{(\mathrm{r})}$ &
$P_{\mathrm{L}}^{(\mathrm{t})}$ &
$\Delta_{\mathrm{r}}$ & $\Delta_{\mathrm{t}}$ &
$\mathcal I=\Delta_{\mathrm{t}}/\Delta_{\mathrm{r}}$ \\
\hline
3 & $3/41$ & $8.20881\times10^{-4}$ & $7.60632\times10^{-4}$ &
$7.03120\times10^{-4}$ & $7.34\ [6.26,8.38]\%$ &
$14.35\ [13.60,15.08]\%$ & $1.95\ [1.72,2.26]$ \\
5 & $7/129$ & $1.39376\times10^{-4}$ & $1.28452\times10^{-4}$ &
$1.18580\times10^{-4}$ & $7.84\ [6.89,8.79]\%$ &
$14.92\ [14.17,15.66]\%$ & $1.90\ [1.71,2.14]$ \\
7 & $13/265$ & $1.89902\times10^{-5}$ & $1.69319\times10^{-5}$ &
$1.56532\times10^{-5}$ & $10.84\ [9.83,11.83]\%$ &
$17.57\ [16.84,18.30]\%$ & $1.62\ [1.49,1.77]$ \\
\end{tabular}
\end{ruledtabular}
\end{table*}

\subsection{Intervention performance across budgets}
\label{app:surface_code_budget_sweep}

In Fig.~\ref{fig:surface_code_budget_sweep}, we compare intervention performance across budgets for the $d=7$ surface code, using seven syndrome rounds, $p=0.001$, and $\alpha=0.5$.
For each $k\in\{1,5,9,13,17,21,25\}$, we select the top $k$ components from the same sensitivity ranking obtained during hotspot identification.
For comparison, we generate 40 distinct random sets for each $k$, each containing $k$ different components chosen uniformly at random from the 265 physical components.
Each circuit is decoded by its own reoptimized PyMatching decoder.

For the circuit without intervention, we use $2.6784\times10^9$ shots for the $X$ memory and $3.0618\times10^9$ for the $Z$ memory, totaling $5.7402\times10^9$ shots.
For hotspot-guided intervention, we use the same numbers of shots for each value of $k$.
The baseline data are shared across all values of $k$.

For random intervention, each of the 40 component sets uses $4.96\times10^7$ shots for the $X$ memory and $5.67\times10^7$ for the $Z$ memory.
The random intervention therefore uses $4.252\times10^9$ shots in total for each value of $k$.
Including the shared baseline once, the simulations across all seven values of $k$ use $7.56856\times10^{10}$ shots in total.

The $k=13$ comparison uses the same simulation data as the $d=7$ comparison in Fig.~\ref{fig:hotspot_targeted_intervention}.

Logical error rates and relative reductions are calculated using the statistical procedure in Appendix~\ref{app:unified_intervention_statistics}.
Table~\ref{tab:surface_code_budget_sweep} reports the relative reductions and their pointwise $95\%$ confidence intervals.

\begin{table*}[t]
\caption{
Relative logical error rate reduction versus the number $k$ of selected components for the $d=7$ rotated surface code at $p=0.001$ and $\alpha=0.5$, corresponding to Fig.~\ref{fig:surface_code_budget_sweep}.
Hotspot-guided intervention selects the $k$ components with the highest sensitivities; random intervention averages the results over 40 sets of $k$ components chosen uniformly at random.
Brackets give $95\%$ bootstrap percentile confidence intervals for individual comparisons.}
\label{tab:surface_code_budget_sweep}
\begin{ruledtabular}
\begin{tabular}{c c c c}
$k$ & Selected components (\%) & Targeted reduction (\%) & Mean random reduction (\%) \\
\hline
1  & $0.38$ & $2.55\ [1.73,3.37]$   & $1.15\ [0.23,2.04]$ \\
5  & $1.89$ & $7.85\ [7.05,8.64]$   & $4.16\ [3.25,5.08]$ \\
9  & $3.40$ & $13.12\ [12.36,13.87]$ & $7.52\ [6.66,8.36]$ \\
13 & $4.91$ & $17.57\ [16.84,18.30]$ & $10.84\ [9.83,11.83]$ \\
17 & $6.42$ & $21.88\ [21.17,22.58]$ & $13.35\ [12.47,14.23]$ \\
21 & $7.92$ & $25.81\ [25.12,26.48]$ & $16.12\ [15.25,16.98]$ \\
25 & $9.43$ & $29.34\ [28.68,29.99]$ & $18.76\ [17.88,19.63]$ \\
\end{tabular}
\end{ruledtabular}
\end{table*}

\section{High-rate lifted product intervention simulations}
\label{app:high_rate_intervention}

\subsection{Code construction and syndrome extraction}

The $[\![150,30,10]\!]$, $[\![200,40,12]\!]$, and $[\![300,60,14]\!]$ codes are mitten codes of Ref.~\cite{mitten2026processors}.  
They are constructed using the lifted product $\mathrm{LP}(A,B)$ of the base matrices $A=\begin{pmatrix}a_0 & a_1\end{pmatrix}$ and $B=\begin{pmatrix}b_0 & b_1\end{pmatrix}$ over $\mathbb{F}_2[G]$, with $G=C_5\times S_3$, $|G|=30$; $G=C_4\times D_{10}$, $|G|=40$; and $G=C_{10}\times S_3$, $|G|=60$, respectively.
Here $C_n$ denotes the cyclic group of order $n$, $S_3$ the symmetric group on three elements (of order $6$), and $D_{10}$ the dihedral group of order $10$.
For $a=\sum_{g\in G}\alpha_g g\in\mathbb{F}_2[G]$, we define the involution $a^*=\sum_{g\in G}\alpha_g g^{-1}$.
Let $\mathbf{e}_h$ denote the standard basis vector indexed by $h\in G$.
The left and right regular representations are defined by $L(g)\mathbf{e}_h=\mathbf{e}_{gh}$ and $R(g)\mathbf{e}_h=\mathbf{e}_{hg^{-1}}$ for $g,h\in G$, and extended linearly to $\mathbb{F}_2[G]$.
With these conventions, $L(a^*)=L(a)^\top$ and $R(a^*)=R(a)^\top$.
The $X$- and $Z$-check matrices of all three codes are
\begin{equation}
\begin{aligned}
    H_X&=\begin{pmatrix}
        L(a_0) & 0 & L(a_1) & 0 & R(b_0^{*})\\
        0 & L(a_0) & 0 & L(a_1) & R(b_1^{*})
    \end{pmatrix},\\[2pt]
    H_Z&=\begin{pmatrix}
        R(b_0) & R(b_1) & 0 & 0 & L(a_0^{*})\\
        0 & 0 & R(b_0) & R(b_1) & L(a_1^{*})
    \end{pmatrix}.
\end{aligned}
    \label{eq:mitten_blocks}
\end{equation}
Since the left and right regular representations commute, the check matrices satisfy $H_XH_Z^{\!\top}=0$.
Each code has $5|G|$ data qubits, $2|G|$ independent $X$ checks, and $2|G|$ independent $Z$ checks.
Every check acts on $9$ data qubits.
Each code therefore encodes $k=5|G|-4|G|=|G|$ logical qubits, giving $k=30,40,60$, respectively, and an encoding rate of $1/5$.

The syndrome extraction circuit is organized according to the block structure of $H_X$ and $H_Z$.
The five block columns in Eq.~\eqref{eq:mitten_blocks} correspond to data qubit blocks $D_1,\ldots,D_5$.
The two block rows of $H_X$ and $H_Z$ correspond to ancilla blocks $X_0,X_1$ and $Z_0,Z_1$, respectively.
Each block contains $|G|$ qubits, labelled by $h\in G$.
We write $D_i(h)$, $X_i(h)$, and $Z_i(h)$ for the qubits labelled $h$ in their respective blocks.
Each of $a_0,a_1,b_0,b_1$ is a sum of three group elements.
For each check type, the qubit pairs associated with the same group-element term of a base entry are coupled in one CNOT layer.
For example, a group element $g$ in $a_0$ gives an $X$-check layer containing the CNOT gates
\begin{equation}
    X_0(h)\to D_1(g^{-1}h),\qquad
    X_1(h)\to D_2(g^{-1}h),
\end{equation}
for all $h\in G$, with all gates applied in parallel.
The three group elements in $a_0$ give three such layers.
The other layers follow in the same way from the blocks of $H_X$ and $H_Z$ in Eq.~\eqref{eq:mitten_blocks}.
The ancillas are the CNOT controls for $X$ checks and the targets for $Z$ checks.
Each check type uses $12$ CNOT layers, three for each of the four base entries, giving $24$ layers per syndrome round.

We reconstruct the syndrome extraction schedules within the family described in Ref.~\cite{mitten2026processors}.
Each schedule consists of four sections, $XL$, $XR$, $ZL$, and $ZR$, where the first letter specifies the check type and the second the left- or right-regular action.
The scheduling choices are the order of these sections and the order of the six CNOT layers within each section.
The schedules used in our simulations may differ from those used in the original study.

For the $[\![150,30,10]\!]$, $[\![200,40,12]\!]$, and $[\![300,60,14]\!]$ codes, the circuits contain $150$, $200$, and $300$ data qubits and $120$, $160$, and $240$ check ancillas, respectively.
The check ancillas are divided equally between $X$ and $Z$ checks.
The corresponding numbers of directed CNOT couplings are $1080$, $1440$, and $2160$.
Counting each qubit and each directed CNOT coupling once across all syndrome rounds gives $1350$, $1800$, and $2700$ physical components, respectively.

\subsection{Noise model and decoding}

For the Mitten-code experiment, we use the same circuit-level depolarizing model as Ref.~\cite{mitten2026processors}.
All noisy operations share a single error probability $p$.
Each qubit initialization is followed by a one-qubit depolarizing channel, and the same channel is applied immediately before each qubit measurement.
Each CNOT is followed by a two-qubit depolarizing channel acting jointly on its control and target qubits.
Both channels are defined by Eq.~\eqref{eq:depolarizing_channel}, with $k=1$ and $k=2$, respectively.
For data qubits, the one-qubit channel is applied after initial preparation and before final measurement.
For ancillas, it is applied after preparation and before measurement in every syndrome round.
No noise is applied during idle periods.

We simulate the $X$- and $Z$-memory experiments at $p=0.002$.
The circuits for the $[\![150,30,10]\!]$, $[\![200,40,12]\!]$, and $[\![300,60,14]\!]$ codes use $10$, $12$, and $14$ rounds of syndrome extraction, respectively.

We use the published telescoping-decoder implementation of Ref.~\cite{mitten2026processors}, with its final integer-programming stage S4 disabled.
For each shot, we first attempt decoding with GPU layered belief propagation (S1).
If this stage does not return an accepted correction, the shot proceeds to GPU relay belief propagation (S2) and, if needed, the CPU serial belief-propagation stages (S3-A/B/C).
Decoding stops at the first stage that returns a correction satisfying its acceptance criteria.
We use the acceptance criteria of Ref.~\cite{mitten2026processors}.
For each code, we keep the decoder parameters fixed across the baseline, hotspot-guided intervention, and random intervention experiments.
We construct the decoder from the detector error model of the circuit being simulated.
For the hotspot-guided and random interventions, this model includes the reduced error probabilities of the selected components.

If none of stages S1--S3 returns an accepted correction, the shot is recorded as a decoder failure.
For shots with an accepted correction, we compare the predicted outcome for each of the $k$ logical qubits with its corresponding simulator outcome in the chosen memory basis.
If any outcome disagrees, the shot is counted as one logical failure.
The reported logical error rate is the fraction of shots with either type of failure.

\subsection{Hotspot identification}

To identify hotspots, we estimate component sensitivities from five independent runs with different random seeds for each code and memory basis.
Summed over all runs in both the $X$ and $Z$ memory bases, the $[\![150,30,10]\!]$ experiment uses $8.15\times10^{7}$ shots and records $4{,}434$ logical failures.
The corresponding totals are $4.00\times10^{8}$ shots and $1{,}884$ logical failures for $[\![200,40,12]\!]$, and $1.775\times10^{9}$ shots and $4{,}038$ logical failures for $[\![300,60,14]\!]$.

Fig.~\ref{fig:mitten_sensitivity_maps} shows the sensitivity of each directed CNOT coupling.
The couplings are arranged by the data qubit blocks, ancilla blocks, and group-element labels introduced above.
For each component, we average the signed sensitivities from the $X$- and $Z$-memory runs with the same random seed and take the absolute value.
We then average these values over the five seeds:
\begin{equation}
    \overline A_c=\frac{1}{5}\sum_s
    \left|\frac{S^X_{c,s}+S^Z_{c,s}}{2}\right|,
    \qquad
    \widetilde A_c=\frac{\overline A_c}{\max_{c'}\overline A_{c'}}.
    \label{eq:mitten_map_score}
\end{equation}
For each code, we normalize the scores by the maximum over all physical components and display them on a $0$--$1$ color scale.

\begin{figure}[!tp]
    \centering
    \includegraphics[width=\columnwidth]{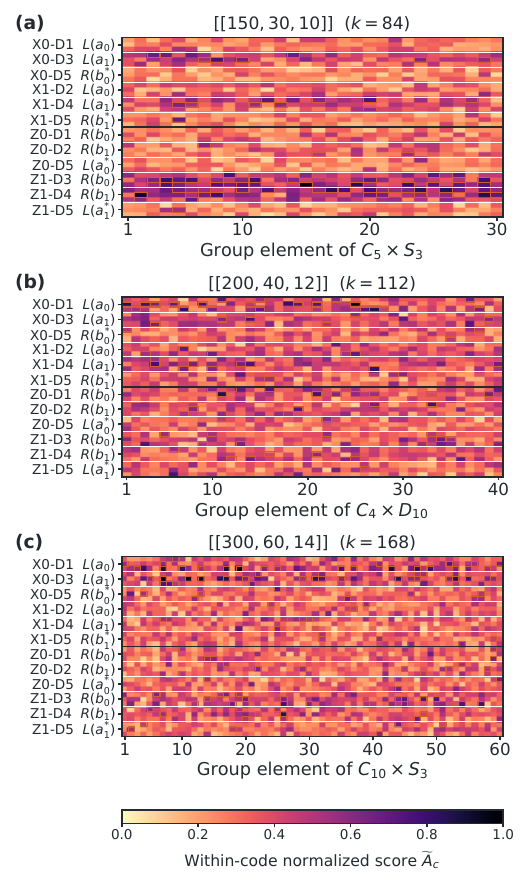}
    \caption{
    Sensitivity maps of directed CNOT couplings for the (a) $[\![150,30,10]\!]$, (b) $[\![200,40,12]\!]$, and (c) $[\![300,60,14]\!]$ mitten codes at $p=0.002$.
    Colors show the relative sensitivities $\widetilde A_c$ defined in Eq.~\eqref{eq:mitten_map_score}, normalized separately for each code.
    Each cell represents one CNOT coupling.
    Rows follow the blocks of $H_X$ and $H_Z$ in Eq.~\eqref{eq:mitten_blocks}, with three rows for the three group-element terms in each nonzero block.
    Columns index the group elements labelling the check ancillas in $C_5\times S_3$, $C_4\times D_{10}$, and $C_{10}\times S_3$, respectively.
    The panels contain all $1080$, $1440$, and $2160$ couplings.
    Orange outlines mark the $77$, $85$, and $132$ couplings selected for intervention.
    The selected sets also contain $7$, $27$, and $36$ check ancillas, which are not shown, giving $k=84$, $112$, and $168$ selected components in total.
    }
    \label{fig:mitten_sensitivity_maps}
\end{figure}

\subsection{Intervention setup and sampling}

The hotspot-guided interventions target $84$, $112$, and $168$ physical components in the $[\![150,30,10]\!]$, $[\![200,40,12]\!]$, and $[\![300,60,14]\!]$ codes, respectively, corresponding to $6.22\%$ of the components in each circuit.
These sets contain no data qubits.
For $[\![150,30,10]\!]$, the selected components comprise $7$ check ancillas and $77$ directed CNOT couplings, accounting for $8.33\%$ and $91.67\%$ of the hotspot set, respectively.
The corresponding counts are $27$ ancillas ($24.11\%$) and $85$ couplings ($75.89\%$) for $[\![200,40,12]\!]$, and $36$ ancillas ($21.43\%$) and $132$ couplings ($78.57\%$) for $[\![300,60,14]\!]$.

The selected CNOT couplings are concentrated in the $L(a_0)$ and $L(a_1)$ blocks of $H_X$ and the $R(b_0)$ and $R(b_1)$ blocks of $H_Z$.
These blocks comprise two thirds of all CNOT couplings but contain $75$ of $77$, $77$ of $85$, and $107$ of $132$ selected couplings in the three codes, respectively.
Specifically, for $[\![150,30,10]\!]$, the two blocks connecting $Z_1$ to $D_3$ and $D_4$ contain $55$ of the $77$ selected couplings ($71.4\%$), while comprising only $16.7\%$ of all CNOT couplings.
For $[\![200,40,12]\!]$, $33$ of the $85$ selected couplings ($38.8\%$) connect to $D_1$: $22$ couple it to $X_0$ and $11$ to $Z_0$.
These two blocks comprise $16.7\%$ of all CNOT couplings.
For $[\![300,60,14]\!]$, $50$ of the $132$ selected couplings ($37.9\%$) connect $X_0$ to $D_1$ or $D_3$, in the $L(a_0)$ and $L(a_1)$ blocks of $H_X$.
These two blocks comprise $16.7\%$ of all CNOT couplings.
Within them, the two group-element rows with the most hotspots, one in each block, contain $33$ selected couplings ($25.0\%$), while comprising only $5.56\%$ of all CNOT couplings.

For code $i$, we denote the selected hotspot set by $H_i$ and fix it before intervention sampling.
The local error rate of each physical component $c\in H_i$ is reduced according to
Eq.~\eqref{eq:targeted_intervention} with $\alpha=0.5$.
For comparison, we simulate random interventions on $40$ component sets for each code.
Each set contains the same number of physical components as $H_i$: $84$, $112$, and $168$ for the three codes, respectively.
As in the hotspot-guided intervention, the local error rate of each selected component is halved.

All shot counts are totals over the $X$ and $Z$ memory bases, with half of the shots in each basis.
For $[\![150,30,10]\!]$, the simulations without intervention and after hotspot-guided intervention each use $5.8\times10^{7}$ shots.
For $[\![200,40,12]\!]$, we use $3.0\times10^{8}$ shots in each case.
For $[\![300,60,14]\!]$, the hotspot-guided intervention uses $1.0556\times10^{9}$ shots.
For its baseline, we reuse the $1.775\times10^{9}$ samples used to estimate component sensitivities.

For random intervention, each of the $40$ component sets uses $1.45\times10^{6}$, $7.5\times10^{6}$, and $2.639\times10^{7}$ shots for the three codes, respectively.
Summed over the $40$ sets, the corresponding totals are $5.8\times10^{7}$, $3.0\times10^{8}$, and $1.0556\times10^{9}$ shots, matching the hotspot-guided intervention for each code.

\subsection{Statistical analysis}

We calculate the logical error rates and confidence intervals using the procedure in Appendix~\ref{app:unified_intervention_statistics}.
Within each memory basis, the logical error rate is the total number of logical failures divided by the total number of shots.
We report the average of the $X$- and $Z$-memory estimates.
For random intervention, we then average these results over the $40$ component sets.

We use $50{,}000$ bootstrap replicates to quantify the statistical uncertainty.
For the simulations without intervention and after hotspot-guided intervention, each replicate independently generates failure counts for each memory basis according to Eq.~\eqref{eq:intervention_bootstrap_counts}.
For random intervention, we draw $40$ values at random from the logical error rate estimates of the $40$ component sets.
We average these $40$ sampled estimates to obtain the logical error rate under random intervention for that bootstrap replicate.
In each replicate, we use the same resampled baseline rate to calculate the reductions from hotspot-guided and random intervention.
The $95\%$ confidence intervals for the logical error rates and their reductions are given by the 2.5th and 97.5th percentiles of the bootstrap estimates.
The confidence intervals for the improvement ratios are calculated using the Fieller method in Eq.~\eqref{eq:intervention_fieller}.

Table~\ref{tab:high_rate_intervention_results} summarizes the mitten code intervention results.

\begin{table*}[t]
\caption{
Logical error rates and relative reductions for the mitten codes at $p=0.002$ and $\alpha=0.5$, corresponding to Fig.~\ref{fig:high_rate_intervention}.
$k$ is the number of selected components, and $F$ and $N$ are the total numbers of logical failures and shots over both memory bases.
For random intervention, these counts also sum over all 40 component sets.
$P_{\mathrm{L}}$ averages the $X$- and $Z$-memory rates and, for random intervention, also averages over the 40 sets.
$\Delta_{\mathrm{t}}$ and $\Delta_{\mathrm{r}}$ are relative reductions to the circuit without intervention; $\mathcal I$ is their ratio.
Brackets give $95\%$ confidence intervals from the $50{,}000$-replicate procedure in Appendix~\ref{app:unified_intervention_statistics}: bootstrap percentile intervals for reductions and Fieller intervals for $\mathcal I$.
}
\label{tab:high_rate_intervention_results}
\begin{ruledtabular}
\setlength{\tabcolsep}{3.1pt}
\begin{tabular}{l c c c c}
Strategy & $F$ & $N$ & $P_{\mathrm{L}}$ &
Relative reduction $[95\%\ \mathrm{CI}]$ \\
\hline
\multicolumn{5}{l}{$[\![150,30,10]\!]$: $k=84$} \\
No intervention & $3{,}184$ & $5.8\times10^{7}$ & $5.48966\times10^{-5}$ &
--- \\
Uniform ($40$ sets) & $2{,}748$ & $5.8\times10^{7}$ &
$4.73793\times10^{-5}$ &
$\Delta_{\mathrm{r}}=13.69\ [9.03,18.17]\%$ \\
Targeted ($k=84$) & $2{,}383$ & $5.8\times10^{7}$ & $4.10862\times10^{-5}$ &
$\Delta_{\mathrm{t}}=25.16\ [21.12,29.03]\%$ \\
\multicolumn{5}{l}{Improvement ratio $\mathcal I=\Delta_{\mathrm{t}}/\Delta_{\mathrm{r}}=1.84\ [1.41,2.65]$} \\
\hline
\multicolumn{5}{l}{$[\![200,40,12]\!]$: $k=112$} \\
No intervention & $1{,}356$ & $3.0\times10^{8}$ & $4.52000\times10^{-6}$ &
--- \\
Uniform ($40$ sets) & $1{,}174$ & $3.0\times10^{8}$ &
$3.91333\times10^{-6}$ &
$\Delta_{\mathrm{r}}=13.42\ [7.02,19.33]\%$ \\
Targeted ($k=112$) & $1{,}044$ & $3.0\times10^{8}$ &
$3.48000\times10^{-6}$ &
$\Delta_{\mathrm{t}}=23.01\ [16.53,29.00]\%$ \\
\multicolumn{5}{l}{Improvement ratio $\mathcal I=\Delta_{\mathrm{t}}/\Delta_{\mathrm{r}}=1.71\ [1.20,2.95]$} \\
\hline
\multicolumn{5}{l}{$[\![300,60,14]\!]$: $k=168$} \\
No intervention & $4{,}038$ & $1.775\times10^{9}$ & $2.27493\times10^{-6}$ &
--- \\
Uniform ($40$ sets) & $2{,}101$ & $1.0556\times10^{9}$ &
$1.99034\times10^{-6}$ &
$\Delta_{\mathrm{r}}=12.51\ [7.39,17.44]\%$ \\
Targeted ($k=168$) & $1{,}797$ & $1.0556\times10^{9}$ & $1.70235\times10^{-6}$ &
$\Delta_{\mathrm{t}}=25.17\ [20.92,29.25]\%$ \\
\multicolumn{5}{l}{Improvement ratio $\mathcal I=\Delta_{\mathrm{t}}/\Delta_{\mathrm{r}}=2.01\ [1.44,3.29]$} \\
\end{tabular}
\end{ruledtabular}
\end{table*}

\clearpage
\bibliography{main}

\end{document}